\documentclass[smallextended]{svjour3}       
\usepackage{cite}
\usepackage[colorlinks]{hyperref}
\hypersetup{
    linkcolor = {blue},
	citecolor = {blue}
}
\smartqed  
\usepackage{url}
\usepackage{bm}
\usepackage{graphicx}
\usepackage{algorithm}
\usepackage[noend]{algorithmic}
\usepackage{subfigure}
\usepackage{amssymb, amsmath,graphicx,charter, latexsym}
\usepackage{enumerate}
\usepackage{epstopdf}
\usepackage{times}      
\usepackage{verbatim} 
\usepackage{booktabs}
\usepackage{mathtools}
\newcommand\givenbase[1][]{\:#1\lvert\:}
\let\given\givenbase
\newcommand\sgiven{\givenbase[\delimsize]}
\DeclarePairedDelimiterX{\Basics}[1]{[}{]}{\let\given\sgiven #1}
\DeclarePairedDelimiterX{\bkt}[1]{(}{)}{\let\given\sgiven #1}
\DeclarePairedDelimiterX{\sbkt}[1]{[}{]}{\let\given\sgiven #1}
\newcommand{\E}{\operatorname{\mathbb{E}} \Basics}

\DeclarePairedDelimiterX{\expectarg}[1]{\left[}{\right]}{%
  \ifnum\currentgrouptype=16 \else\begingroup\fi
  \activatebar#1
  \ifnum\currentgrouptype=16 \else\endgroup\fi
}
\newcommand{\innermid}{\nonscript\;\delimsize\vert\nonscript\;}
\newcommand{\activatebar}{%
  \begingroup\lccode`\~=`\|
  \lowercase{\endgroup\let~}\innermid 
  \mathcode`|=\string"8000
}
\newcommand{\bigO}{\ensuremath{\Omega}}
\DeclareMathOperator{\Tr}{Tr}

\DeclareMathOperator{\T}{\mathbb{T}}

\newcommand{\RN}[1]{%
  \textup{\uppercase\expandafter{\romannumeral#1}}%
}
\usepackage{color}
\newcommand{\bs}{\mathbf}
\begin{document}

\title{Delay-Optimal Scheduling for Queueing Systems with Switching Overhead}
\author{Ping-Chun Hsieh \and I-Hong Hou \and Xi Liu}

\institute{Ping-Chun Hsieh \and I-Hong Hou \and Xi Liu \at
Department of ECE, Texas A\&M University, College Station, Texas 77843-3128, USA. \\
\email: \{pingchun.hsieh, ihou, xiliu\}@tamu.edu
}
\date{Received: date / Accepted: date}
\maketitle

\begin{abstract}
We study the scheduling polices for asymptotically optimal delay in queueing systems with switching overhead. Such systems consist of a single server that serves multiple queues, and some capacity is lost whenever the server switches to serve a different set of queues. The capacity loss due to this switching overhead can be significant in many emerging applications, and needs to be explicitly addressed in the design of scheduling policies. For example, in 60GHz wireless networks with directional antennas, base stations need to train and reconfigure their beam patterns whenever they switch from one client to another. Considerable switching overhead can also be observed in many other queueing systems such as transportation networks and manufacturing systems.
While the celebrated Max-Weight policy achieves asymptotically optimal average delay for systems without switching overhead, it fails to preserve throughput-optimality, let alone delay-optimality, when switching overhead is taken into account.
We propose a class of Biased Max-Weight scheduling policies that explicitly takes switching overhead into account. The Biased Max-Weight policy can use either queue length or head-of-line waiting time as an indicator of the system status. We prove that our policies not only are throughput-optimal, but also can be made arbitrarily close to the asymptotic lower bound on average delay.
To validate the performance of the proposed policies, we provide extensive simulation with various system topologies and different traffic patterns. 
We show that the proposed policies indeed achieve much better delay performance than that of the state-of-the-art policy.
\end{abstract}

\section{Introduction}
\label{section: introduction}
Design of scheduling policies is one of the most critical parts in achieving good performance for queueing systems. Based on the system information, such as queue backlog or instantaneous service rate, the server dynamically switches between different scheduling decisions. The delay required for each switch is traditionally assumed to be minimal compared to the service time of each job, and therefore can be neglected. However, for applications that require dynamic tuning or strict safety guarantees, this switching overhead has to be explicitly addressed with caution. Wireless networking on the 60 GHz band \cite{Fan2008,Singh2011,Niu2015}, for example, is inherently featured by the switching overhead. Signal propagation on the 60 GHz band suffers from much more serious attenuation than on the widely-used 2.4 GHz or 5 GHz spectrum. To cope with this increased attenuation, beamforming techniques based on directional antennas have been widely applied to mitigate this problem \cite{Nitsche2014}. However, as stated in \cite{Navda2007}, to reconfigure the beam direction, the beam switching procedure can take up to several hundred microseconds, which is about the transmission time of a packet. While this switching overhead is usually overlooked in most wireless communication networks, this beam-switching latency needs to be explicitly addressed in directional-antenna applications. 

Another example is the traffic signal control for signalized intersections. During each signal phase transition, the traffic signal controller has to undergo a yellow-to-red period and an all-red period to guarantee safety. Without proper scheduling, the switching overhead can greatly reduce the intersection capacity \cite{Ghavami2012}. Existing studies have shown that the conventional fixed time scheduling policy can result in a significant amount of capacity loss. \cite{Allsop1972,Varaiya2013}. Moreover, this switching overhead is expected to be even larger in mixed transportation networks with both human-driven and autonomous vehicles \cite{LeVine2015,Liu2015}. Recently, there has been growing interest in designing scheduling algorithms to achieve maximum traffic throughput in transportation networks \cite{Varaiya2013,Wongpiromsarn2012}. Meanwhile, more traffic-responsive scheduling policies are now applicable due to the recent progress in both vehicle-to-infrastructure and vehicle-to-vehicle communication. Indeed, to approach maximum throughput in real traffic scenarios, the influence of switching overhead on network capacity has to be incorporated and overcome in scheduling design.

The effect of switching overhead is also critical in many other applications, such as multi-thread processors \cite{David2007}, passive optical networks \cite{Mcgarry2008}, and manufacturing systems \cite{Sharifnia1991}. As a classic example, polling system is one widely studied queueing model that incorporates switching overhead. In such systems, the server serve the queues in cyclic order, and the server requires finite time to go from one queue to the next queue. The polling system model has been applied in a variety of applications, such as communication networks \cite{Hanbali2008} and manufacturing systems \cite{Sharifnia1991}. The survey of the polling model and summary of the theoretical results can be found in \cite{Takagi1988,Levy1990,Takagi1997,Vishnevskii2006}. Despite the progress in the study of the polling model, there are still a vast variety of applications that require multiple queues to be served simultaneously while the polling model assumes only one queue is served at each time.

For queueing systems that allow simultaneous service of multiple queues, a class of scheduling policies named the Max-Weight policies, introduced by Tassiulas and Ephremides in \cite{Tassiulas1992}, has been shown to achieve optimal throughput for the multi-hop queueing systems. Later on, there are several extended works of the Max-Weight policies, such as \cite{Eryilmaz2005,Tassiulas1993}. In addition to throughput-optimality, the Max-Weight policy has also been shown to achieve good delay performance in different queueing systems \cite{Eryilmaz2012,Neely2009,Kar2012,Le2009,Gupta2010}. To obtain delay bound for the Max-Weight scheduling, one common technique is to set the drift of a Lyapunov function to zero, as in \cite{Neely2009,Kar2012,Le2009,Gupta2010,Eryilmaz2012}. Besides, Eryilmaz and Srikant \cite{Eryilmaz2012} derives an asymptotically tight upper bound on queue length for the Max-Weight policy. However, none of the above literature considers switching overhead incurred by the transition between different schedules. In fact, Max-Weight policy fails to preserve throughput-optimality when the switching overhead is considered. The reason is that Max-Weight policy may suffer from excessive switching and therefore waste a significant portion of time on switching.


To remedy the instability issue of Max-Weight policy, Armony and Bambos \cite{Armony2003} propose cone policies and batch policies and prove that the two families of policies are throughput-optimal with non-zero switching overhead.
Later on, Hung and Chang \cite{Hung2008} propose an extended version of Dynamic Cone policy to reduce the complexity and improve delay performance of the conventional cone policies.
Chan \emph{et al.} \cite{Chan2016} proposes the Maximum-Weight-Matching with hysteresis (MWM-H) policy to achieve optimal throughput with deterministic service rates as well as to properly distribute the queue backlog when the queues become overloaded. Besides, Celik \emph{et al.} \cite{Celik2016} propose the Variable Frame-Based Max-Weight (VFMW) policy which incorporate frame structure to avoid frequent switching, and show that the policy preserves throughput-optimality with nonzero switching overhead for interference networks. 
However, all of the above studies focus only on throughput-optimality and do not provide any theoretical bound on delay performance. 
In this paper, we regard the VFMW policy in \cite{Celik2016} as the state-of-the-art policy and the VFMW policy serves as a reference for the comparison of delay performance. 
For the delay analysis of queueing systems with switching overhead, Perkins and Kumar \cite{Perkins1989} propose a class of exhaustive policies to achieve asymptotically optimal delay for single-server queueing systems with deterministic arrival and service processes. However, it is not even clear whether the exhaustive-type policies can also achieve optimal throughput in systems where multiple queues can be served simultaneously. 

In this paper, we propose a class of scheduling policies called \emph{Biased Max-Weight} (BMW) policies to achieve both throughput-optimality and delay-optimality in the presence of switching overhead. The weight of each schedule can be defined by either the queue backlog or the waiting time of the head-of-line (HOL) job. Like the Max-Weight policy, the BMW policy still makes scheduling decisions based on the weight of each schedule, but it has a bias towards the current schedule. In other words, the server makes a switch only when the new schedule is better than the previous schedule by an enough margin. The BMW policies share similar design philosophy with the VFMW policy and the MWM-H policy to avoid excessive switching according to the current queue status. Different from those prior works, we are able to characterize the queue length statistics and therefore the average delay. To achieve this goal, we first prove that the BMW scheduling policy is throughput-optimal by showing that the underlying Markov chain is positive recurrent. We further prove that the BMW policy achieves the asymptotically tight queue length bound by choosing parameter used by BMW arbitrarily close to zero. Through extensive simulation, we show that the BMW policy still provides good delay performance when the server may serve multiple queues simultaneously, subject to some conflicting constraints imposed by the system.

In addition to average delay, the per-queue delay is also crucial for many applications, such as transportation networks. We further compare the queue-length-based BMW (Q-BMW) and the waiting-time-based BMW (W-BMW) in terms of per-queue average delay. Under the Q-BMW policy, each queue has about the same average queue length. By Little's law, the queue with a smaller mean arrival rate is expected to experience a larger average delay. On the other hand, the W-BMW policy follows the same switching scheme as the Q-BMW policy, but with the HOL waiting times as the state variables. Since the HOL waiting time serves as a good indicator of per-queue average delay, the W-BMW can avoid holding up the jobs for too long in the queues with lighter incoming traffic. Simulation results show that the W-BMW policy indeed achieves better fairness than the Q-BMW policy without sacrificing system-wide average delay.

The rest of the paper is organized as follows. Section \ref{section: model} describes the system model as well as the notations in use. Section \ref{section: prelim} provides mathematical preliminaries of queue stability and delay-optimality and also describes the fundamental dilemma in achieving delay-optimality. Section \ref{section: throughput-optimal Q-BMW} describes the throughput-optimality of the Q-BMW policy. Then, Section \ref{section: delay-optimal Q-BMW} provides an asymptotic upper bound on average delay under Q-BMW policy. For the W-BMW policy, Section \ref{section: W-BMW} shows that it achieves throughput-optimality and the same delay upper bound as the Q-BMW policy. Section \ref{section: simulation} provides extensive simulation results of queueing networks with different constraints. Finally, Section \ref{section: conclusion} concludes the paper.  



\section{Notation and System Models}
\label{section: model}
\subsection{Network Model}
We consider a time-slotted switched queueing system with one centralized server and $N\in \mathbb{N}_{0}$ ($\mathbb{N}_{0}$ is a shorthand for $\mathbb{N}\cup \{0\}$) parallel queues, which are indexed by $\mathcal{N}=\{1,2,...,N\}$. Time slots are indexed by $t\in \mathbb{N}_{0}$. Each queue is associated with an exogenous traffic stream. Arriving jobs are first buffered in the queue and leave the system right after the service is completed. The server may be able to serve multiple queues simultaneously. A set of queues that can be served simultaneously is called a \emph{feasible schedule}.
We represent each feasible schedule by an $N$-dimensional binary vector $\bs{I}=(I_1, I_2, \cdots, I_N)$, where $I_i$ is the indicator function of whether queue $i$ is included in the schedule. Throughout the paper, we focus on \emph{maximal} feasible schedules: a feasible schedule is \emph{maximal} if no additional queue can be added to the schedule. With a little abuse of notation, we use $\lvert \bs{I}\rvert$ to denote the number of queues included in the schedule $\bs{I}$. Suppose there are $J$ maximal feasible schedules denoted by $ \mathcal{I}=\{\bs{I}^{(1)},\cdots, \bs{I}^{(J)}\}$. 
In each time slot $t$, the server selects a maximal feasible schedule $\bs{I}(t)\in \mathcal{I}$ according to its scheduling policy $\eta$. 

When the server switches from one schedule to serving another schedule, it needs to spend $T_s\in \mathbb{N}$ slots on preparing for the transition before working on the new schedule. The delay $T_s$ reflects the switching overhead which is usually overlooked in the queueing systems but needs to be explicitly addressed in many applications, such as directional-antenna systems as well as transportation systems described in Section \ref{section: introduction}. Therefore, there are two operation modes for the server: {ACTIVE} mode and {SWITCH} mode. We let $M(t)$ be the indicator function of whether the server is in ACTIVE mode at time $t$. We use $t_k$ to denote the time when the server makes a switch for the $k$-th time, and set $t_0=0$. The time between two consecutive switches is called an \emph{interval}. Let $T_k:=t_{k+1}-t_{k}$ denote the length of the $k$-th interval, for all $k\in \mathbb{N}_0$. Note that $T_k$ reflects how frequently the server is switching between different schedules.

Throughout the paper, we use $(\mathcal{N}, \mathcal{I}, T_s)$ to denote a queueing system described in the above.

\subsection{Traffic Model}


We model the arrival process $\{A_i(t)\}_{t}$ of each queue $i$ by a sequence of independent and identically distributed (i.i.d.) non-negative random variables $A_i(t)\in \mathbb{N}_0$ with $\E{A_i(t)}=\lambda_i$ for all $t\geq 0$.
We further assume that $A_i(t)$ is upper bounded, i.e. there exists a finite constant $A_{\max}>0$ such that $A_i(t)\leq A_{\max}$ for every $t\geq 0$. Similarly, the service process $\{S_i(t)\}_{t}$ of each queue $i$ is modeled by a sequence of i.i.d. non-negative random variable $S_i(t)\in \mathbb{N}_0$ with $\E{S_i(t)}=\mu_i$ and $S_i(t)\leq S_{\max}$, for all $t\geq 0$. We also assume that the server does not collect any information about the instantaneous service rates. The reason is that with non-zero switching overhead the server is not able to exploit the time-varying service when the service rates are independent across time. Moreover, let $\rho_i:=\lambda_i/\mu_i$ be the \emph{normalized traffic load} of queue $i$. For every queue $i$, each arriving job is labeled with a unique sequence number. For each queue $i$, the sequence numbers start from 1 and would indicate the arrival order of the jobs. At each time $t$, we use $\varphi_{i}(t)$ to denote the sequence number of the latest completed job of queue $i$ before time $t$. Let $\{V_{i}(m)\}_{m\geq 0}$ be the inter-arrival time process of queue $i$, where $V_{i}(m)$ denote the inter-arrival time between the two jobs with sequence numbers $m$ and $m+1$ in queue $i$. 

To simplify notations in later sections, we use boldface letters $\bs{A}, \bs{S}, \bm{\lambda}, \bm{\mu}$ and $\bm{\rho}$ to denote the $N$-dimensional vectors of the arrivals, services, mean arrival rates, mean service rates, and normalized traffic loads, respectively. We also use $\lambda_{\max}$, $\lambda_{\min}$, $\mu_{\max}$ and $\mu_{\min}$ as the shorthands of the maximum mean arrival rate, minimum mean arrival rate, maximum mean service rate, and minimum mean service rate among all the queues, respectively.

\subsection{Queue Dynamics}
Let $Q_i(t)\in \mathbb{N}_{0}$ be the number of jobs buffered in queue $i$ at time $t$. We assume $Q_i(0)=0$ for every queue $i$. Define $\hat{S}_i(t):=\min\{Q_i(t), M(t)I_i(t)S_i(t)\} $ to be the amount of service actually used by each queue $i$ at time $t$.
Throughout this paper, we consider the store-and-forward queueing model, i.e., for each queue $i$,
\begin{equation}
Q_i(t+1) = Q_i(t) - \hat{S}_i(t) + A_i(t), \hspace{6pt} \forall t\geq 0.
\end{equation}
For simplicity, we use $\bs{Q}(t)=(Q_1(t), ..., Q_N(t))\in \mathbb{N}_{0}^{N}$ to denote the queue length status of the system at time $t$. Moreover, in our model the queue length process $\{\bs{Q}(t)\}_{t\geq 0}$ form a discrete-time Markov chain on a countable state space $\mathbb{N}_{0}^{N}$. The total queue length of the system can be written as $\bs{1}^{T}\bs{Q}(t)$, where $\bs{1}^T$ denotes the $N$-dimensional all-ones row vector.

Let $W_i(t)\in \mathbb{N}_{0}$ be the waiting time of the head-of-line (HOL) job of queue $i$ and $\bs{W}(t)=(W_1(t),...,W_N(t))\in \mathbb{N}_{0}^{N}$ be the corresponding HOL waiting time vector. 
Then, the HOL waiting time can be updated as follows:
\begin{equation}
W_i(t+1)= \max\left\{0, \bkt[\bigg]{W_i(t)-\sum_{j=1}^{\hat{S}_i(t)} V_{i}(\varphi(t)+j)}\right\}.
\end{equation}
Similar to the queue length process, the HOL waiting time process $\{\bs{W}(t)\}_{t\geq 0}$ also form a discrete-time Markov chain on a countable state space $\mathbb{N}_{0}^{N}$. 

\section{Preliminaries}
\label{section: prelim}
\subsection{Capacity Region and a Lower-Bound for Delay}
In preparation for the discussion of delay performance, we first introduce the fundamental concepts of queue stability, throughput-optimality, and delay-optimality. 
First, we formally introduce one commonly used definition of queue stability.
\begin{definition}(\emph{Strong stability})
The queueing system is \emph{strongly stable} under a scheduling policy if
\begin{equation}
\limsup_{t\rightarrow \infty}\frac{\sum_{\tau=0}^{t-1}\E[\big]{\bs{1}^T \bs{Q}(\tau)}}{t}< \infty. \hspace{3pt}\qed
\end{equation}
\end{definition}

Based on the above definitions, we can classify the arrival rate vectors by queue stability and define the capacity region.
\begin{definition}(\emph{Admissible arrival rates})
An arrival rate vector $\bs{\lambda}=(\lambda_1,..., \lambda_N)$ is said to be \emph{admissible} if there exists a scheduling policy under which the queueing system is strongly stable.  $\qed$
\end{definition}

\begin{definition}(\emph{Capacity region})
The \emph{capacity region} $\mathrm{\Lambda}\subset \mathbb{R}_{+}^{N}$ of the system is defined as the closure of the set that consists of all the admissible arrival rate vectors.  $\qed$
\end{definition}

The following lemma shows that the capacity region can be fully characterized by the normalized traffic load vector and the maximal feasible schedules.
\begin{lemma}(Characterization of capacity region)
For any queueing system described in Section \ref{section: model}, given the mean service rate vector $\bm{\mu}$, the capacity region can be characterized as
\begin{equation}
\mathrm{\Lambda} = \left\{\bm{\lambda} \hspace{3pt} \middle\lvert \hspace{3pt}\exists \bm{\beta}\geq 0 \hspace{3pt}\textrm{with}\hspace{3pt} \sum_{j=1}^{J}\beta_{j}\leq 1 \hspace{3pt}\textrm{such that}\hspace{3pt} \bm{\rho}\leq \sum_{j=1}^{J}\beta_{j}\bs{I}^{(j)}\right\}. \hspace{3pt} \qed
\end{equation}
\end{lemma}
\begin{proof}
This is a direct result of Theorem 1 in \cite{Celik2012}. $\qed$
\end{proof}

Here, we introduce the notion of throughput-optimality:
\begin{definition}(\emph{Throughput-optimality})
A scheduling policy $\eta$ is said to be \emph{throughput-optimal} if for any interior point $\bm{\lambda}$ of $\mathrm{\Lambda}$, the system is strongly stable under $\eta$.  $\qed$
\end{definition}

Given $\bm{\lambda}$, $\bm{\mu}$, and $\mathcal{I}$, we can further describe the traffic load of the whole system. 
\begin{definition}(\emph{Utilization factor})
Given the queueing system described in Section \ref{section: model}, we define the \emph{utilization factor} as
\begin{equation}
\beta^*:=\min_{\bm{\beta}:\sum_{j=1}^{J}\beta_j \bs{I}^{j}\geq \bm{\rho}} { \bs{1}^T\bm{\beta}}. \label{equation:epsilon star}
\end{equation}
For convenience, we also define $\epsilon^*:=1-\beta^*$, which reflects the "distance" from the boundary of the capacity region. $\qed$
\end{definition}

From the study in \cite{Eryilmaz2012}, the average delay of a queueing system is closely related to $\epsilon^{*}$. A useful lower bound for the average delay is provided here as an easy reference.
\begin{lemma}(Lower bound on queue length without switching overhead)
\label{lemma: general delay lower bound without Ts}
Given a stable queueing system $\mathcal{Q}=(\mathcal{N}, \mathcal{I}, T_s)$ described in Section \ref{section: model} with queue length process $\{\bs{Q}(t)\}_{t}$ and $T_s = 0$, under any scheduling policy $\eta$, the expected total queue length in steady state scales as 
\begin{equation}
\label{equation: general delay lower bound}
\E{\bs{1}^{T}\bs{Q}(t)}=\mathrm{\Omega}(1/\epsilon^{*}). \qed
\end{equation} 
\end{lemma}
\begin{proof}
This is a direct result of Lemma 6 in \cite{Eryilmaz2012}. $\qed$
\end{proof}

\begin{remark}
By Little's law, (\ref{equation: general delay lower bound}) implies that the total average delay also scales as $\mathrm{\Omega}(1/\epsilon^{*})$ for systems without switching overhead. 
\end{remark}
\begin{remark}
In \cite{Eryilmaz2012}, the lower bound (\ref{equation: general delay lower bound}) is obtained by constructing a hypothetical single-queue system from the original multi-queue system, whose total queue length is larger in stochastic ordering than that of the constructed single-queue system. Following this argument, we can also derive a similar lower bound on expected queue length for the queueing systems with switching overhead.
\end{remark}

\begin{corollary}(Lower bound on queue length with switching overhead)
\label{corollary: delay lower bound with overhead}
Given a stable queueing system $\mathcal{Q}=(\mathcal{N}, \mathcal{I}, T_s)$ described in Section \ref{section: model} with queue length process $\{\bs{Q}(t)\}_{t}$ and $T_s > 0$, under any scheduling policy $\eta$, the expected total queue length in steady state is 
\begin{equation}
\E{\bs{1}^{T}\bs{Q}(t)}=\mathrm{\Omega}(1/\epsilon^{*}). \hspace{3pt} \qed
\end{equation} 
\end{corollary}
\begin{proof}
Given a queueing system $\mathcal{Q}$, by following the same procedure as in Lemma 6 of \cite{Eryilmaz2012}, we can construct a single-queue system $\mathcal{Q}^{'}$ with switching overhead $T_s$. Then we know the queue length process of $\mathcal{Q}$ is larger in stochastic ordering than the queue length process of $\mathcal{Q}^{'}$. Next, we can construct another single-queue system $\mathcal{Q}^{''}$ from $\mathcal{Q}^{'}$ but with $T_s=0$. Then we know the queue length process of $\mathcal{Q}^{'}$ is larger in stochastic ordering than the queue length process of $\mathcal{Q}^{''}$.
By Lemma \ref{lemma: general delay lower bound without Ts}, we have $\E{\bs{1}^{T}\bs{Q}(t)}=\mathrm{\Omega}(1/\epsilon^{*})$. $\qed$
\end{proof}

Based on Corollary \ref{corollary: delay lower bound with overhead} , we can define asymptotic delay-optimality as follows.
\begin{definition}(\emph{Delay-optimality})
A scheduling policy is \emph{delay-optimal} if in steady state, the total average queue length satisfies that
\begin{equation}
\E{\bs{1}^{T}\bs{Q}(t)}=\textrm{O}(1/\epsilon^{*}).
\end{equation}
In other words, $\textrm{O}(1/\epsilon^{*})$ is an asymptotically tight delay upper bound. $ \qed$
\end{definition}


\subsection{The Variable-Frame Max-Weight Scheduling Policy}
\label{subsection: VFMW}
As discussed in Section \ref{section: introduction}, despite the progress in throughput-optimal scheduling for systems with switching overhead, it is still not clear how to achieve optimal delay performance for such queueing systems with stochastic arrival and service processes. In the prior work \cite{Celik2016}, the Variable-Frame Max-Weight (VFMW) policy has been proposed  to achieve throughput optimality queueing systems with switching overhead. Under the VFMW policy, we need to determine a sublinear function for calculating frame size such that the server stays with the same schedule till the end of the frame. While the frame size function has no effect on throughput-optimality (as long as it is sublinear), it can indeed greatly affect the delay performance. In the example discussed in \cite{Celik2016}, the frame function is chosen to be $\bkt[\big]{\sum_{i}Q_i(t)}^{\alpha}$, where $\alpha$ is between 0 and 1. Besides, \cite{Celik2016} also suggests that $\alpha$ should be chosen as close to 1 as possible based on their simulation results. However, the $\alpha$ value with the smallest delay can actually differ in different scenarios. This phenomenon can be easily observed through simulation as follows.

For example, we consider a single-sever system of 4 queues with Bernoulli arrival and service processes. The switching overhead $T_s=1$. We provide simulation results of two scenarios with different mean arrival rates and mean service rates in Figure \ref{figure:VFMW power topology 1} and \ref{figure:VFMW power topology 2}. 
\begin{itemize}
\item Scenario \RN{1}: $\bm{\lambda}=(0.119, 0.119, 0.119, 0.119)$, $\bm{\mu}=(0.5, 0.5, 0.5, 0.5)$
\item Scenario \RN{2}: $\bm{\lambda}=(0.08, 0.25, 0.09, 0.01)$, $\bm{\mu}=(0.8, 0.5, 0.3, 0.2)$
\end{itemize}
Note that the utilization factor of both scenarios is 0.95. In Scenario \RN{1}, the smallest total delay is achieved when $\alpha$ is close to 1. However, the optimal $\alpha$ is about 0.6 in Scenario \RN{2}. This example demonstrates that there does not exist a fixed value of $\alpha$ that can achieve the optimal delay performance for the VFMW policy. We study the trace files of our simulations and find that the VFMW policy can suffer from poor delay performance from two conflicting factors:
\begin{itemize}
\item If $\alpha$ is large, say close to 1, then the frame size grows fast. The VFMW policy may stick to an inefficient feasible schedule for too long and thereby suffers from large delay. 
\item If $\alpha$ is small, say close to 0, then the frame size is very small for most of the time and becomes very insensitive to the change in queue status. Consequently, the VFMW policy may switch too frequently and gets severely impacted by the switching overhead. 
\end{itemize}

The above arguments highlight the fundamental difficulty in achieving delay-optimality for queueing systems with switching overhead: If a policy switches too frequently, it suffers from too much capacity loss due to switching overhead. On the other hand, if a policy does not switch often enough, it may stay with a schedule that is no longer efficient for too long. In the next section, we propose our online scheduling policy that solves such dilemma.

\begin{figure}[!htbp]
\centering
\subfigure[Scenario I.]{
\includegraphics[scale=0.4]{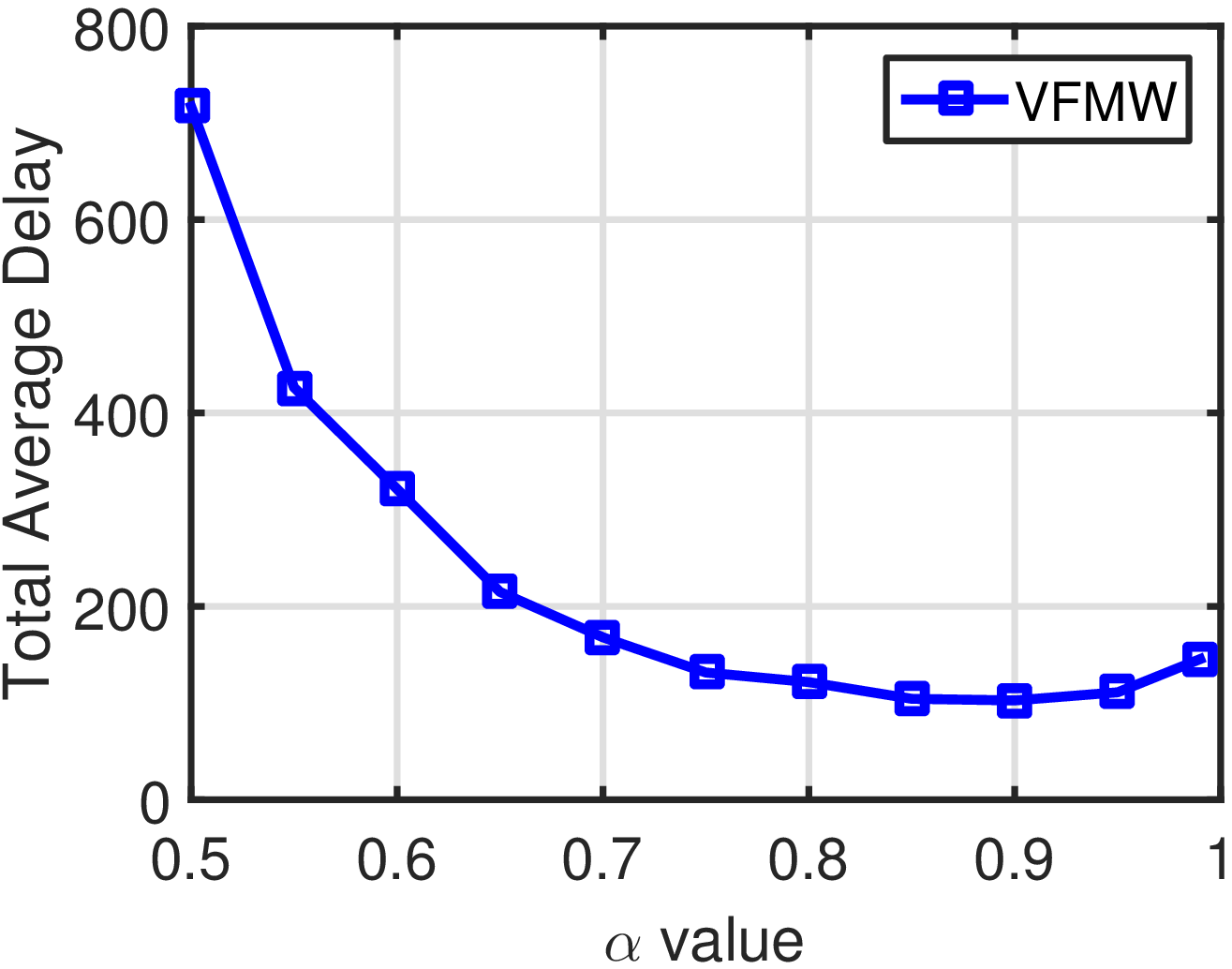}
\label{figure:VFMW power topology 1}}
\hspace{5mm}
\subfigure[Scenario II.]{
\includegraphics[scale=0.4]{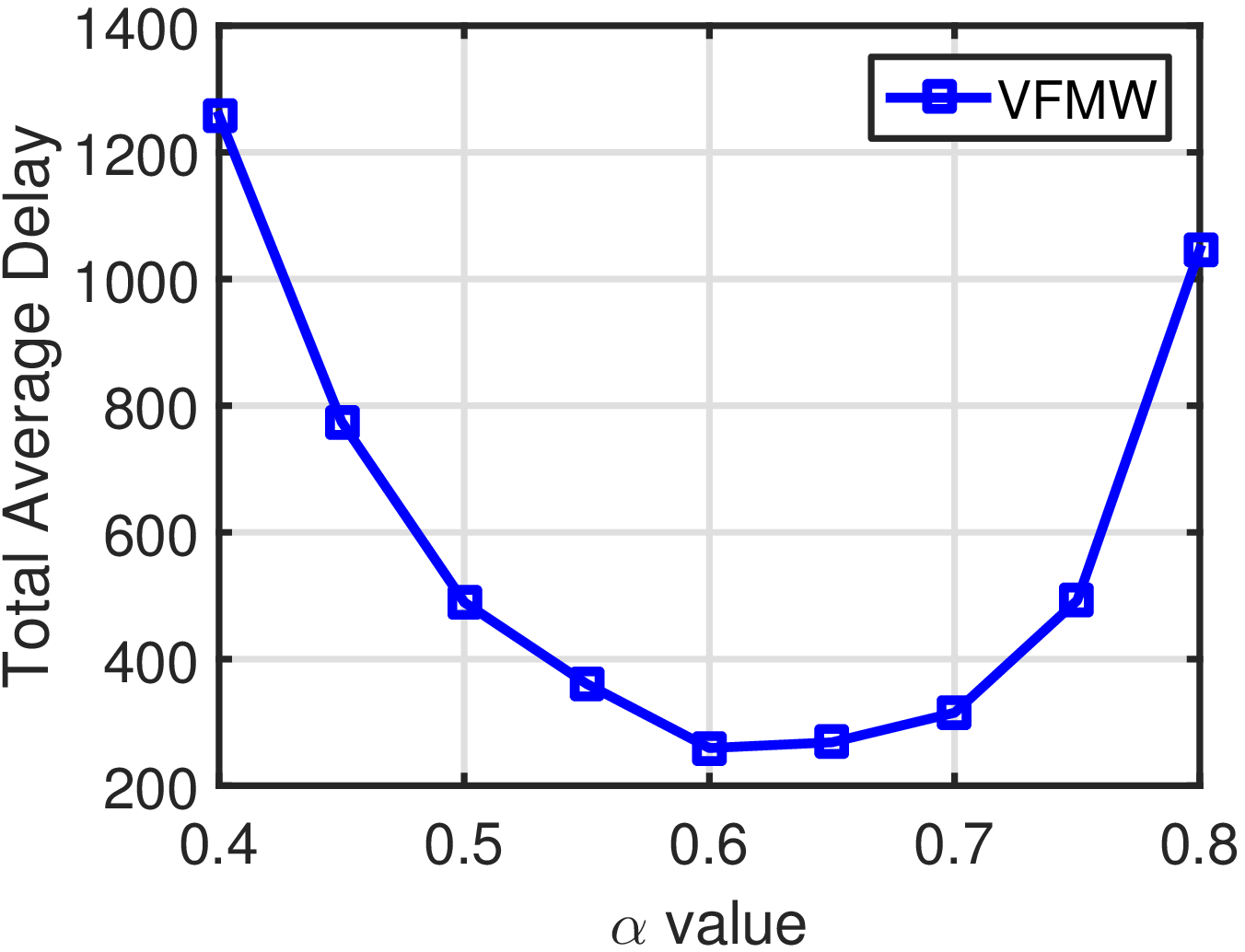}
\label{figure:VFMW power topology 2}}
\caption{Average total delay versus different $\alpha$ value under VFMW policy with frame function $\bkt[\big]{\sum_{i}Q_i(t)}^{\alpha}$.}
\end{figure}

\section{Throughput-Optimal Scheduling for Systems With Switching Overhead}
\label{section: throughput-optimal Q-BMW}
\subsection{Scheduling Policy and Intuitive Description}
We propose two types of \emph{Biased Max-Weight scheduling policy} to achieve both throughput-optimality and the asymptotically tight delay bound. We first introduce the queue-length-based Biased Max-Weight (Q-BMW) policy as follows. \newline

\noindent {\bf Q-BMW policy:} Let $F(\cdot): \mathbb{R}_{+}^N\rightarrow [1,\infty)$ be a function chosen by the server. At each time $t$ in the $k$-th interval, if the system satisfies that
\begin{equation}
\label{equation:policy}
\bkt[\bigg]{1+\frac{T_s}{F(\bs{Q}(t_k))}}\bkt[\Big]{\bs{I}(t_k)^T \bs{Q}(t)} \leq \bkt[\bigg]{\max_{j: 1\leq j \leq J}(\bs{I}^{(j)})^T \bs{Q}(t)},
\end{equation}
then the server makes a switch to serve the schedule with the largest sum of queue lengths at time $t$ (ties are broken arbitrarily).
Otherwise, the server stays with the current schedule. $\qed$\newline

Note that the Q-BMW policy does not rely on the frame structure adopted by the VFMW policy. Instead, the Q-BMW policy overcomes the switching overhead by giving an intentional bias to the current schedule. Moreover, the Q-BMW keeps checking the condition of (\ref{equation:policy}) in each time slot. Intuitively, the Q-BMW policy avoids the dilemma highlighted in the previous section because:
\begin{itemize}
\item Since the Q-BMW policy checks (\ref{equation:policy}) in each time slot, it cannot stick to an inefficient schedule for too long, regardless of the choice of the function $F(\cdot)$.
\item Since the Q-BMW favors the current schedule, it can avoid switching too frequently as long as the bias to the current schedule is not too small. This suggests that one should choose a function $F(\cdot)$ that increases very slowly with $\bs{Q}(t)$.
\end{itemize}

\subsection{Throughput-Optimality of the Q-BMW Scheduling}
To show that the Q-BMW is throughput-optimal, we first introduce a lower bound on $T_k$ in the following lemma.
\begin{lemma}
\label{lemma:lower bound for Tk}
Suppose the server can serve at most $K$ queues at a time. Under the Q-BMW policy, for any $k\geq 0$ and for every sample path, in the $k$-th interval we have
\begin{equation}
T_k\geq {C_0(\bs{1}^T\bs{Q}(t_k))}/{F(\bs{Q}(t_k))},
\end{equation}
where $C_0=T_s/\bkt[\big]{NK\bkt[\big]{A_{\max}+(1+T_s)S_{\max}}}$.   $\qed$
\end{lemma}
\begin{proof}
Suppose at time $t=t_k+\tau$, $\tau>0$ in the $k$-th interval, the server enters SWITCH mode and starts switching. Then, there exists some schedule $\bs{I}^{(m)}\neq \bs{I}(t_k)$ such that
\begin{equation}
\bkt[\bigg]{1+\frac{T_s}{F(\bs{Q}(t_k))}}\bkt[\Big]{\bs{I}(t_k)^T \bs{Q}(t_k+\tau)} \leq \bkt[\Big]{(\bs{I}^{(m)})^T \bs{Q}(t_k+\tau)}.\label{equation:Tk lower bound 1}
\end{equation}
Moreover, by the boundedness of the arrival processes, we know
\begin{equation}
\bkt[\big]{\bs{I}^{(m)}}^T \bkt[\Big]{\bs{Q}(t_k)+\tau A_{\max} \bs{1}}\geq \bkt[\big]{\bs{I}^{(m)}}^T\bs{Q}(t_k+\tau).\label{equation:Tk lower bound 2} 
\end{equation}
From (\ref{equation:Tk lower bound 1}) and (\ref{equation:Tk lower bound 2}), we have
\begin{align}
\bkt[\big]{\bs{I}^{(m)}}^T \bkt[\Big]{\bs{Q}(t_k)+\tau A_{\max} \bs{1}}&\geq \bkt[\Big]{\bs{I}(t_k)^T \bs{Q}(t_k+\tau)} \bkt[\bigg]{1+\frac{T_s}{F(\bs{Q}(t_k))}}\\
& \geq \bkt[\Big]{\bs{I}(t_k)^T \bkt[\big]{\bs{Q}(t_k)-\tau S_{\max}\bs{1}}} \bkt[\bigg]{1+\frac{T_s}{F(\bs{Q}(t_k))}}
\end{align}
Next, we rearrange the above equations as
\begin{align}
K\bkt[\Big]{A_{\max}+(1+T_s) S_{\max}}\tau&\geq \bs{I}(t_k)^T\bs{Q}(t_k)-\bkt[\big]{\bs{I}^{(m)}}^T\bs{Q}(t_k)+\frac{T_s{\bs{I}(t_k)}^T\bs{Q}(t_k)}{F(\bs{Q}(t_k))}\\
&\geq \frac{T_s{\bs{I}(t_k)}^T\bs{Q}(t_k)}{F(\bs{Q}(t_k))}\\
&\geq \frac{T_s{\bs{1}}^T\bs{Q}(t_k)}{N\cdot F(\bs{Q}(t_k))}
\end{align}
Hence, we can get the lower bound on $T_k$ as
\begin{equation}
\label{equation: Q-BMW Tk lower bound}
T_k\geq \frac{T_s\cdot{\bs{1}}^T\bs{Q}(t_k)}{NK \bkt[\big]{A_{\max}+(1+T_s)S_{\max}}\cdot F(\bs{Q}(t_k))}.
\end{equation}
$\qed$
\end{proof}

Lemma \ref{lemma:lower bound for Tk} provides important insight in choosing a proper $F(\bs{Q}(t_k))$. We now state the main theorem of throughput-optimality as follows.

\begin{theorem}
\label{theorem: Q-BMW throughput-optimal}
If we choose $F(\bs{Q}(t_k))=\max\{1, (\bs{1}^T \bs{Q}(t_k))^{\alpha}\}$ with $\alpha \in(0,1)$, then the Q-BMW policy is throughput-optimal. Moreover, the underlying Markov chain induced by the queue length process $\{\bs{Q}(t)\}_{t\geq 0}$ is positive recurrent. $\qed$
\end{theorem}
\begin{proof}
The complete proof is provided in Appendix 1. In summary, for any system with mean arrival rates vector $\bm{\lambda}$ in the interior of the capacity region $\mathrm{\Lambda}$, we show that the queueing system is strongly stable by applying the Lyapunov drift framework.  That is, we utilize a quadratic Lyapunov function and show that the expected Lyapunov drift is negative. Different from the one-step drift which is often used in the queueing systems without switching overhead, we choose a hypothetical observation window and show that the multi-step Lyapunov drift across the window is negative. Based on the lower bound on $T_k$ given by (\ref{equation: Q-BMW Tk lower bound}), the effect of switching overhead is amortized over the whole observation window. $\qed$
\end{proof}

\section{Asymptotically Tight Queue Length Bound under Q-BMW Scheduling}
\label{section: delay-optimal Q-BMW}
In this section, we focus on systems where the server can serve at most one queue at a time, that is, $\lvert \bs{I}\rvert=1$, for all feasible schedule $\bs{I}$. We show that Q-BMW is nearly delay-optimal when $\alpha\downarrow 0$ by proving the following theorem:

\begin{theorem}
\label{theorem:single-channel delay-optimal stochastic}
For any queueing system $\mathcal{Q}=(\mathcal{N}, \mathcal{I}, T_s)$ described in Section \ref{section: model} where the server can serve at most one queue at a time, the Q-BMW scheduling policy provides the following queue length upper bound: there exists some constant $B<\infty$ such that
\begin{equation}
\lim_{\epsilon^{*}\downarrow 0}\epsilon^{*}\E[\Big]{(\bs{1}^{T}\bs{Q}(t))^{1-\alpha}}\leq B.
\end{equation}
Hence, $\E[\big]{(\bs{1}^{T}\bs{Q}(t))^{1-\alpha}}$ scales as $\textrm{O}(1/\epsilon^{*})$.
By choosing $\alpha$ arbitrarily close to 0, the Q-BMW policy achieves the asymptotically tight queue length bound. $\qed$
\end{theorem}

We first introduce some necessary definitions and lemmas for the proof of Theorem \ref{theorem:single-channel delay-optimal stochastic}.

\begin{definition}
A scheduling policy $\eta$ is said to be \emph{work-conserving} if the server never serves an empty queue whenever there is an unfinished job in the system. $\qed$
\end{definition}

\begin{definition}
A scheduling policy $\eta$ is said to be \emph{ergodic} if the Markov chain resulting from $\eta$ is positive recurrent. $\qed$
\end{definition}

\begin{lemma}
\label{lemma:single-channel work-conserving Tk upper bound}
Let $K_{\T}$ be the number of intervals in $[0,\T)$. For any queueing system under any ergodic work-conserving policy $\eta$, there exists some constant $B_0<\infty$ such that 
\begin{equation}
\limsup_{\T\rightarrow \infty}\frac{\sum_{k=1}^{K_{\T}}T_k}{K_{\T}}\leq \frac{B_{0}}{\epsilon^{*}},
\end{equation}
almost surely. $\qed$
\end{lemma}
\begin{proof}
Let the counting process $Y_i(t)$ be the number of slots for which queue $i$ is served up to $t$. Moreover, let $\tau_{i,1}, \tau_{i,2},...,\tau_{i,Y_i(t)}$ be these time slots. 
Let $Z_i(t)$ denote the cumulative time for which the scheduled service for queue $i$ is not fully utilized up to time $t$, that is,
\begin{equation}
Z_i(t):=\sum_{\tau=0}^{t-1} \mathbb{I}_{\{\hat{S}_i(\tau)>Q_i(\tau)\}},
\end{equation}
where $\mathbb{I}_{\{\cdot\}}$ denotes the indicator function of an event. At each time $\tau$, the event that $\{\hat{S}_i(\tau)>Q_i(\tau)\}$ happens only when the queue becomes empty at the beginning of time $\tau+1$. Since the policy $\eta$ is work-conserving, at the beginning of slot $\tau+1$ the server will switch to a new schedule. Therefore, the total cumulative time for which the scheduled service is not fully utilized is at most the same as the number of intervals, that is, for any $\mathbb{T}>0$,
\begin{equation}
\sum_{i=1}^{N} Z_i(\mathbb{T}) \leq K_{\mathbb{T}},\label{equation:sum of Zi less than KT}
\end{equation}
where $K_{\mathbb{T}}$ is the number of intervals in $[0,\mathbb{T})$. Since the policy $\eta$ is work-conserving, we also have
\begin{equation}
\sum_{i=1}^{N}Y_i(\T)\geq \T-K_{\T}T_s.\label{equation:time conservation}
\end{equation}
Moreover, we have for each queue $i$, 
\begin{equation}
\label{equation:service conservation}
S_{\max}Z_i({\T})\geq \sum_{m=1}^{Y_i(\T)}S_i(\tau_{i,m})-\sum_{\tau=0}^{\T-1}A_i(\tau).
\end{equation}
The right-hand side of (\ref{equation:service conservation}) represents the cumulative service that is not fully utilized by queue $i$ up to $\T$. 
After dividing both sides of (\ref{equation:service conservation}) by $\mathbb{T}$, we have
\begin{align}
\frac{S_{\max}Z_i(\mathbb{T})}{\mathbb{T}} &\geq \frac{\sum_{m=1}^{Y_i(\mathbb{T})}S_i(\tau_{i,m})}{\mathbb{T}} - \frac{\sum_{\tau=0}^{\mathbb{T}-1}A_i(\tau)}{\mathbb{T}}\\ \label{equation:Z greater than S minus A}
& = \frac{Y_i(\mathbb{T})}{\mathbb{T}}\frac{\sum_{m=1}^{Y_i(\mathbb{T})}S_i(\tau_{i,m})}{Y_i(\mathbb{T})}-\frac{\sum_{\tau=0}^{\mathbb{T}-1}A_i(\tau)}{\mathbb{T}}
\end{align}
Since the policy $\eta$ is ergodic and thereby the Markov chain resulting from $\eta$ is positive recurrent, then $\lim_{\T \rightarrow \infty} \frac{Y_i(\T)}{\T}$ exists, for every queue $i$. By letting $\mathbb{T}\rightarrow \infty$, we have
\begin{align}
\liminf_{\mathbb{T}\rightarrow \infty}\frac{S_{\max}Z_i(\mathbb{T})}{\mathbb{T}} &\geq\liminf_{\mathbb{T}\rightarrow\infty}\frac{Y_i(\mathbb{T})}{\mathbb{T}}\frac{\sum_{m=1}^{Y_i(\mathbb{T})}S_i(\tau_{i,m})}{Y_i(\mathbb{T})} - \limsup_{\mathbb{T}\rightarrow\infty}\frac{\sum_{\tau=0}^{\mathbb{T}-1}A_i(\tau)}{\mathbb{T}}. \label{equation:limit of Z greater than S minus A}
\end{align}
Note that $\lim_{\mathbb{T}\rightarrow \infty}Y_i(\mathbb{T})\rightarrow\infty$ since the queue $i$ cannot be stable if otherwise. 
By the Strong Law of Large Numbers, we have $\lim_{\mathbb{T}\rightarrow\infty}\frac{\sum_{\tau=0}^{\mathbb{T}-1}A_i(\tau)}{\mathbb{T}}=\lambda_i$ and $\lim_{\mathbb{T}\rightarrow \infty}\frac{\sum_{m=1}^{Y_i(\mathbb{T})}S_i(\tau_{i,m})}{Y_i(\mathbb{T})}=\mu_i$, for every queue $i$.
Therefore, (\ref{equation:limit of Z greater than S minus A}) can be written as 
\begin{equation}
\frac{S_{\max}}{\mu_i}\cdot\liminf_{\mathbb{T}\rightarrow \infty}\frac{Z_i(\mathbb{T})}{\mathbb{T}} \geq \lim_{\mathbb{T}\rightarrow \infty} \frac{Y_i(\mathbb{T})}{\mathbb{T}} - \rho_{i}. \label{equation:limit of Z greater than Y minus rho }
\end{equation}
By summing (\ref{equation:limit of Z greater than Y minus rho }) over all $i$, we have
\begin{align}
\frac{NS_{\max}}{\mu_{\min}} \cdot\liminf_{\mathbb{T}\rightarrow \infty}\frac{\sum_{i=1}^{N}Z_{i}(\mathbb{T})}{\mathbb{T}}&\geq \lim_{\mathbb{T}\rightarrow \infty}\frac{\sum_{i=1}^{N}Y_i(\mathbb{T})}{\mathbb{T}}-\sum_{i=1}^{N}\rho_{i}\\
\label{equation: limit inf KT over T}
&\geq 1-\liminf_{\T\rightarrow \infty}\frac{K_{\T}T_{s}}{\T}-\sum_{i=1}^{N}\rho_{i}\\ 
&=\epsilon^{*}-\liminf_{\T\rightarrow \infty}\frac{K_{\T}T_{s}}{\T},
\end{align}
where (\ref{equation: limit inf KT over T}) holds from the inequality (\ref{equation:time conservation}).
By using (\ref{equation:sum of Zi less than KT}), we then have
\begin{equation}
\frac{NS_{\max}}{\mu_{\min}}\cdot \liminf_{\mathbb{T}\rightarrow \infty}\frac{K_{\T}}{\T}\geq \epsilon^{*}-\liminf_{\T\rightarrow \infty}\frac{K_{\T}T_{s}}{\T}.
\end{equation}
Therefore, we obtain that $\liminf_{\T\rightarrow \infty}\frac{K_{\T}}{\T}\geq  {\epsilon^{*}}\bkt[\big]{T_s+\frac{NS_{\max}}{\mu_{\min}}}^{-1}$, almost surely.
Equivalently, we have
\begin{equation}
\limsup_{\T\rightarrow\infty} \frac{\sum_{k=1}^{K_{\T}}T_k}{K_{\T}}\leq \frac{T_s + \frac{NS_{\max}}{\mu_{\min}}}{\epsilon^{*}}, 
\end{equation}
almost surely. $\qed$
\end{proof}

\begin{lemma}
\label{lemma:single-channel Q-BMW work-conserving}
For any queueing system $\mathcal{Q}=(\mathcal{N}, \mathcal{I}, T_s)$ described in Section \ref{section: model} where the server can serve at most one queue at a time, the Q-BMW scheduling policy is work-conserving. $\qed$
\end{lemma}
\begin{proof}
Under the Q-BMW policy, if the scheduled queue becomes empty, that is, $\bs{I}(t)^T\bs{Q}(t)=0$, and if there still exists another non-empty queue, then the switching condition (\ref{equation:policy}) should be triggered. Therefore, the Q-BMW policy never idles when there is still a job in the system. $\qed$
\end{proof}

\begin{theorem}
\label{theorem:single-channel Q-BMW Tk}
For any queueing system $\mathcal{Q}=(\mathcal{N}, \mathcal{I}, T_s)$ described in Section \ref{section: model} where the server can serve at most one queue at a time, under the Q-BMW scheduling policy, there exists some constant $B_0<\infty$ such that 
\begin{equation}
\limsup_{\T\rightarrow\infty}\frac{\sum_{k=1}^{K_{\T}}T_k}{K_{\T}}\leq \frac{B_0}{\epsilon^{*}},
\end{equation}
almost surely. Moreover, if the system is strongly stable and therefore the underlying Markov chain is positive recurrent, then we also have
\begin{equation}
\lim_{\epsilon^{*}\downarrow 0} \epsilon^{*} \E{T_k}\leq B_0.  
\end{equation}
$\qed$
\end{theorem}
\begin{proof}
This is a direct result of Lemma \ref{lemma:single-channel work-conserving Tk upper bound} and Lemma \ref{lemma:single-channel Q-BMW work-conserving}. $\qed$
\end{proof}

The following lemma shows that to derive the queue length bound in steady state, we can consider only the queue length at the beginning of each interval. 
\begin{lemma}
\label{lemma: queue length bound for all t}
Given $\gamma\in (0,1]$, in steady state, if there exists some positive constant $B_{0}<\infty$ such that at the beginning of any interval
\begin{equation}
\lim_{\epsilon^{*}\downarrow 0}\epsilon^{*}\E[\Big]{(\bs{1}^{T}\bs{Q}(t_k))^{\gamma}}\leq B_{0},
\end{equation}
then there also exists a positive constant $B_{1}<\infty$ such that in any time slot $t$
\begin{equation}
\lim_{\epsilon^{*}\downarrow 0}\epsilon^{*}\E[\Big]{(\bs{1}^{T}\bs{Q}(t))^{\gamma}}\leq B_{1}.
\end{equation}
$\qed$
\end{lemma}
\begin{proof}
For any time slot $t$ in the $k$-th interval, we have
\begin{align}
\E[\Big]{\bkt[\big]{\bs{1}^{T}\bs{Q}(t)}^{\gamma}}&\leq \E[\bigg]{\bkt[\Big]{\bs{1}^{T}\bs{Q}(t_k) + (t-t_k)\sum_{i=1}^{N}A_{\max}}^{\gamma}}\\
&\leq \E[\bigg]{\bkt[\big]{(\bs{1}^{T}\bs{Q}(t_k)}^{\gamma} + (t-t_k)\sum_{i=1}^{N}A_{\max}}\\
&\leq \E[\Big]{\bkt[\big]{(\bs{1}^{T}\bs{Q}(t_k)}^{\gamma} + N A_{\max}T_k}.
\end{align}
Therefore, by Theorem \ref{theorem:single-channel Q-BMW Tk}, we know that
\begin{align}
\lim_{\epsilon^{*}\downarrow 0}\epsilon^{*} \E[\Big]{\bkt[\big]{\bs{1}^{T}\bs{Q}(t)}^{\gamma}} &\leq \lim_{\epsilon^{*}\downarrow 0} \epsilon^{*}\E[\Big]{\bkt[\big]{\bs{1}^{T}\bs{Q}(t_k)}^{\gamma}} + \lim_{\epsilon^{*}\downarrow 0} \epsilon^{*}\E[\big]{NA_{\max}T_k}<\infty.
\end{align}
This completes the proof. $\qed$
\end{proof}

We are now ready to prove Theorem \ref{theorem:single-channel delay-optimal stochastic}.
\begin{proof} (Theorem \ref{theorem:single-channel delay-optimal stochastic})
Given ${\bs{Q}}(t_k)$, for any $\tau\geq T_s$, define $\Delta \widetilde{\bs{Q}}(t_k+\tau):=\bs{Q}(t_k+\tau)-\bkt[\Big]{\bs{Q}(t_k)+\tau\bm{\lambda}-(\tau-T_s)\bkt[\big]{\bm{\mu}\circ{\bs{I}(t_k)}}}$, where $\bm{\mu}\circ {\bs{I}(t_k)}$ denotes the element-wise product of $\bm{\mu}$ and ${\bs{I}(t_k)}$. Note that $\Delta \widetilde{\bs{Q}}(t_k+\tau)$ represents the "deviation" in queue backlog with stochastic arrival and service processes from that with deterministic arrival rates and service rates. 
Therefore, at time $t_{k+1}$, under the Q-BMW policy,
\begin{align}
\label{equation: Q-BMW deterministic estimate 1}
&\bkt[\bigg]{1+\frac{T_s}{F(\bs{Q}(t_k))}}\bkt[\Big]{{\bs{I}(t_k)^T}\bkt[\Big]{\bs{Q}(t_k)+T_k\bm{\lambda}- (T_k-T_s)\bm{\mu}+\Delta \widetilde{\bs{Q}}(t_{k+1})}}\\
\label{equation: Q-BMW deterministic estimate 2}
&\leq \bkt[\bigg]{\bs{I}(t_{k+1})^T \bkt[\Big]{\bs{Q}(t_k)+T_k \bm{\lambda}+\Delta \widetilde{\bs{Q}}(t_{k+1}))}}.
\end{align}
Since $\bs{I}(t_k)^{T}\bs{Q}(t_k)\geq \bs{I}(t_{k+1})^{T}\bs{Q}(t_k)$, (\ref{equation: Q-BMW deterministic estimate 1}) and (\ref{equation: Q-BMW deterministic estimate 2}) can be rearranged as
\begin{align}
\frac{T_s\bs{I}(t_k)^T \bs{Q}(t_k)}{F(\bs{Q}(t_k))}&\leq \bkt[\bigg]{1+\frac{T_s}{F(\bs{Q}(t_k))}}\bkt[\Bigg]{\bkt[\Big]{T_k \bs{I}(t_k)^T\bkt[\big]{\bm{\mu}-\bm{\lambda}}} - \bs{I}(t_k)^T\Delta \widetilde{\bs{Q}}(t_{k+1})}\\
&\hspace{72pt}+ T_k\bs{I}(t_{k+1})^T\bm{\lambda}+\bs{I}(t_{k+1})^T\Delta \widetilde{\bs{Q}}(t_{k+1})\\
&\leq T_k \bkt[\Big]{\mu_{\max}(T_s+1)+\lambda_{\max}}+(T_s+2)\sum_{i=1}^{N}{\lvert\Delta\widetilde{\bs{Q}}_i(t_{k+1})\rvert}.
\end{align}
By the Functional Law of Iterated Logarithm \cite{Chen2001}, with probability one we have
\begin{equation}
\Delta\widetilde{{Q}}_{i}(t_{k+1}) = \textrm{O}(\sqrt{T_k \log \log T_k}), \hspace{12pt} \forall i=1,...,N.
\end{equation}
Therefore, by choosing $F(\bs{Q}(t_k))$ as in Theorem \ref{theorem: Q-BMW throughput-optimal}, we have
\begin{equation}
\bkt[\bigg]{\sum_{k=1}^{K_{\T}}{\bkt[\Big]{\bs{1}^T \bs{Q}(t_k)}^{1-\alpha}}}\leq N{\sum_{k=1}^{K_{\T}}\bkt[\Big]{ \bkt[\big]{\mu_{\max}(T_s+1)+\lambda_{\max}}\frac{T_k}{T_s} + \textrm{O}(\sqrt{T_k \log \log T_k})}}. \label{equation:single-beam queue length bound stochastic}
\end{equation}
By dividing both sides of (\ref{equation:single-beam queue length bound stochastic}) by $K_{\T}$ and using Theorem \ref{theorem:single-channel Q-BMW Tk}, we know there exists some constant $B_0<\infty$ such that
\begin{equation}
\lim_{{\T}\rightarrow \infty}\frac{\sum_{k=1}^{K_{\T}}{\bkt[\Big]{\bs{1}^T \bs{Q}(t_k)}^{1-\alpha}}}{K_{\T}}\leq{\frac{B_0}{\epsilon^{*}}},
\label{equation:single-beam queue length bound stochastic 2}
\end{equation}
almost surely. For any $\alpha\in (0,1)$, by Theorem \ref{theorem: Q-BMW throughput-optimal}, we know that the Markov chain induced by $\{Q(t)\}_{t\geq 0}$ is positive recurrent and therefore
\begin{equation}
\label{equation: limiting average to expectation}
\E[\Big]{\bkt[\big]{{\bs{1}^T \bs{Q}(t_k)}}^{1-\alpha}}= \lim_{\T\rightarrow \infty}\frac{\sum_{k=1}^{K_{\T}} \bkt[\big]{{\bs{1}^T \bs{Q}(t_k)}}^{1-\alpha}}{K_{\T}},
\end{equation} 
almost surely. Hence, by Lemma \ref{lemma: queue length bound for all t} along with (\ref{equation:single-beam queue length bound stochastic 2}) and (\ref{equation: limiting average to expectation}), there exists a positive constant $B<\infty$ such that
\begin{equation}
\lim_{\epsilon^{*}\downarrow 0}\epsilon^{*}\E[\Big]{(\bs{1}^{T}\bs{Q}(t))^{1-\alpha}}\leq B.
\end{equation}
By choosing $\alpha$ arbitrarily close to 0, the Q-BMW policy indeed achieves the asymptotically tight queue length bound. $\qed$
\end{proof}

\section{Waiting-Time-Based Biased Max-Weight Scheduling}
\label{section: W-BMW}
\subsection{Throughput-Optimality}
We extend the framework introduced in Section \ref{section: throughput-optimal Q-BMW} and Section \ref{section: delay-optimal Q-BMW} to the waiting-time-based Biased Max-Weight (W-BMW) scheduling policy. Throughout this section, we relax the assumption that the arrival processes are i.i.d. for all the queues. Instead, we make a mild assumption on the arrival processes: for each queue $i$, the inter-arrival times $\{V_i(m)\}_{m\geq 0}$ form an i.i.d. sequence and are upper bounded by a constant $V_{\max}<\infty$, almost surely. Note that with this assumption, the following analysis of W-BMW policy also applies to queueing systems with periodic arrivals.
 
\vspace{2mm}
\noindent {\bf W-BMW policy:} Let $G(\cdot): \mathbb{R}_{+}^N\rightarrow [1,\infty)$ be a function chosen by the server. At each time $t$ in the $k$-th interval, if the system satisfies
\begin{equation}
\label{equation:HOL policy}
\bkt[\bigg]{1+\frac{T_s}{G(\bs{W}(t_k))}}\bkt[\Big]{\bs{I}(t_k)^T \bs{W}(t)} \leq \bkt[\bigg]{\max_{j: 1\leq j \leq J}(\bs{I}^{(j)})^T \bs{W}(t)},
\end{equation}
then the server enters SWITCH mode to prepare for serving the schedule with the largest sum of HOL waiting time at time $t$ (ties are broken arbitrarily).
Otherwise, the server stays with the current schedule. $\qed$\newline

\begin{remark}
Both the Q-BMW and W-BMW have the same performance in terms of throughput-optimality and delay-optimality defined in Section \ref{section: prelim}. The advantage of W-BMW is that it achieves better fairness than the Q-BMW policy in terms of per-queue average delay, especially when there is a large difference in the arrival rates and service rates between different queues. We will further describe this feature of W-BMW through simulation in Section \ref{section: simulation}.
\end{remark}
\begin{lemma}
\label{lemma:HOL lower bound for Tk}
Suppose the server can serve at most $K$ queues at a time. Under the W-BMW policy, for every $k\geq 0$ and for every sample path, in the $k$-th interval we have
\begin{equation}
T_k\geq {C_1(\bs{1}^T\bs{W}(t_k))}/{F(\bs{W}(t_k))},
\end{equation}
where $C_1=T_s/\bkt[\big]{NK\bkt[\big]{1+(1+T_s)S_{\max}V_{\max}}}$.  $\qed$
\end{lemma}
\begin{proof}
Suppose at time $t=t_k+\tau$, $\tau>0$, the server enters SWITCH mode and starts switching. Then, there exists some schedule $\bs{I}^{(m)}\neq \bs{I}(t_k)$ such that
\begin{equation}
\bkt[\bigg]{1+\frac{T_s}{G(\bs{W}(t_k))}} \bkt[\Big]{\bs{I}(t_k)^T \bs{W}(t_k+\tau)} \leq \bkt[\Big]{(\bs{I}^{(m)})^T \bs{W}(t_k+\tau)}.\label{equation:HOL Tk lower bound 1}
\end{equation}
Moreover, we know
\begin{equation}
\bkt[\big]{\bs{I}^{(m)}}^T \bkt[\Big]{\bs{W}(t_k)+\tau \bs{1}}\geq \bkt[\big]{\bs{I}^{(m)}}^T\bs{W}(t_k+\tau).\label{equation:HOL Tk lower bound 2} 
\end{equation}
From (\ref{equation:HOL Tk lower bound 1}) and (\ref{equation:HOL Tk lower bound 2}), we have
\begin{align}
\bkt[\big]{\bs{I}^{(m)}}^T \bkt[\Big]{\bs{W}(t_k)+\tau \bs{1}}&\geq \bkt[\Big]{\bs{I}(t_k)^T \bs{W}(t_k+\tau)} \bkt[\bigg]{1+\frac{T_s}{G(\bs{W}(t_k))}}\\
& \geq \bkt[\Big]{\bs{I}(t_k)^T \bkt[\big]{\bs{W}(t_k)-\tau S_{\max}V_{\max}\bs{1}}} \bkt[\bigg]{1+\frac{T_s}{G(\bs{W}(t_k))}}
\end{align}
Next, we rearrange the above equations as
\begin{align}
K\bkt[\Big]{1+(1+T_s) S_{\max}V_{\max}}\tau&\geq \bs{I}(t_k)^T\bs{W}(t_k)-\bkt[\big]{\bs{I}^{(m)}}^T\bs{W}(t_k)+\frac{T_s{\bs{I}(t_k)}^T\bs{W}(t_k)}{G(\bs{W}(t_k))}\\
&\geq \frac{T_s{\bs{I}(t_k)}^T\bs{W}(t_k)}{G(\bs{W}(t_k))}\\
&\geq \frac{T_s{\bs{1}}^T\bs{W}(t_k)}{N\cdot G(\bs{W}(t_k))}
\end{align}
Hence, we can get the lower bound on $T_k$:
\begin{equation}
T_k\geq\frac{T_s\cdot{\bs{1}}^T\bs{W}(t_k)}{NK\bkt[\big]{1+(1+T_s)S_{\max}V_{\max}}\cdot G(\bs{W}(t_k))}.
\end{equation}
$\qed$
\end{proof}

Next, we show that W-BMW policy is also throughput-optimal in the following theorem.
\begin{theorem}
\label{theorem: W-BMW throughput-optimal}
If we choose $G(\bs{W}(t_k))=\max\{1, (\bs{1}^T \bs{W}(t_k))^{\alpha}\}$ with $\alpha\in(0,1)$, then the W-BMW policy is throughput-optimal. Moreover, the underlying Markov chain induced by the waiting time process $\{\bs{W}(t)\}_{t\geq 0}$ is positive recurrent. $\qed$
\end{theorem}
\begin{proof}
The proof is provided in Appendix 2. We use the similar technique as in the proof of Theorem \ref{theorem: Q-BMW throughput-optimal} to prove that W-BMW is throughput-optimal. $\qed$
\end{proof}

\subsection{Asymptotically Tight Queue Length Bound under the W-BMW Scheduling}
As in Section \ref{section: delay-optimal Q-BMW}, we focus on systems where the server can serve at most one queue at a time. We show that W-BMW is also nearly delay-optimal when $\alpha \downarrow 0$ by proving the following theorem:

\begin{theorem}
\label{theorem:HOL single-channel delay-optimal stochastic}
Suppose the server can serve at most one queue at a time. For any such queueing system $\mathcal{Q}=(\mathcal{N}, \mathcal{I}, T_s)$ and stochastic arrival and service processes as described in Section \ref{section: model} and Section \ref{section: W-BMW}, W-BMW policy provides the following upper bound on queue length: there exists some constant $B<\infty$ such that
\begin{equation}
\lim_{\epsilon^{*}\downarrow 0}\epsilon^{*}\E[\Big]{(\bs{1}^{T}\bs{Q}(t))^{1-\alpha}}\leq B.
\end{equation}
Hence, $\E[\big]{(\bs{1}^{T}\bs{Q}(t))^{1-\alpha}}$ scales as $\textrm{O}(1/\epsilon^{*})$. By choosing $\alpha$ to be arbitrarily close to 0, the W-BMW policy achieves asymptotically tight queue length bound and hence it is delay-optimal. $\qed$
\end{theorem}

We introduce some necessary lemmas for the proof of Theorem \ref{theorem:HOL single-channel delay-optimal stochastic}.

\begin{lemma}
\label{lemma:W-BMW single-channel work-conserving}
Suppose the server can serve at most one queue at a time. For any queueing system $\mathcal{Q}=(\mathcal{N}, \mathcal{I}, T_s)$ described in Section \ref{section: model}, the W-BMW policy is work-conserving. $\qed$
\end{lemma}
\begin{proof}
By definition, $Q_i(t) = 0$ implies $W_i(t) = 0$ for any queue $i$ and any time $t$. Under the W-BMW policy, if the scheduled queue becomes empty, then we have $\bs{I}(t)^T\bs{W}(t)=0$. Meanwhile, if there also exists another non-empty queue, then the switching condition (\ref{equation:HOL policy}) should be triggered. Therefore, the W-BMW policy never idles when there is still a job in the system. $\qed$
\end{proof}

\begin{theorem}
\label{theorem:W-BMW single-channel Tk}
Let $K_{\T}$ be the number of intervals in $[0,\T)$. For any queueing system $\mathcal{Q}=(\mathcal{N}, \mathcal{I}, T_s)$ described in Section \ref{section: model} where the server can serve at most one queue at a time, under the W-BMW policy, there exists some constant $B_0<\infty$ such that
\begin{equation}
\lim_{\T\rightarrow\infty}\frac{\sum_{k=1}^{K_{\T}}T_k}{K_{\T}}\leq \frac{B_0}{\epsilon^{*}},
\end{equation}
almost surely. $\qed$
\end{theorem}
\begin{proof}
This is a direct result of Lemma \ref{lemma:single-channel work-conserving Tk upper bound} and Lemma \ref{lemma:W-BMW single-channel work-conserving}. $\qed$
\end{proof}

We are now ready to prove Theorem \ref{theorem:HOL single-channel delay-optimal stochastic} as follows.

\begin{proof} (Theorem \ref{theorem:HOL single-channel delay-optimal stochastic})
Given ${\bs{W}}(t_k)$, for any $\tau\geq T_s$, define $\Delta \widetilde{\bs{W}}(t_k+\tau):=\bs{W}(t_k+\tau)-\bkt[\Big]{\bs{W}(t_k)+\tau\bs{1}-(\tau-T_s)\bkt[\big]{\bm{\rho}^{-1}\circ{\bs{I}(t_k)}}}$, where $\bm{\rho}^{-1}:=(\rho_{1}^{-1},...,\rho_{N}^{-1})$ and $\bm{\rho}^{-1}\circ {\bs{I}(t_k)}$ denotes the element-wise product (also called Hadamard product) of $\bm{\rho}^{-1}$ and ${\bs{I}(t_k)}$. 
Therefore, at time $t_{k+1}$, under the W-BMW policy, we have
\begin{align}
\label{equation: W-BMW deterministic estimate 1}
&\bkt[\bigg]{1+\frac{T_s}{G(\bs{W}(t_k))}}\bkt[\Big]{{\bs{I}(t_k)^T}\bkt[\Big]{\bs{W}(t_k)+T_k\bs{1}- (T_k-T_s)\bkt[\big]{\bm{\rho}^{-1}\circ{\bs{I}(t_k)}}+\Delta \widetilde{\bs{W}}(t_{k+1})}}\\
\label{equation: W-BMW deterministic estimate 2}
&\leq \bkt[\bigg]{\bs{I}(t_{k+1})^T \bkt[\Big]{\bs{W}(t_k)+T_k\bs{1}+\Delta \widetilde{\bs{W}}(t_{k+1}))}}
\end{align}
Since $\bs{I}(t_k)^{T}\bs{W}(t_k)\geq \bs{I}(t_{k+1})^{T}\bs{W}(t_k)$, we can rearrange (\ref{equation: W-BMW deterministic estimate 1}) and (\ref{equation: W-BMW deterministic estimate 2}) as
\begin{align}
\frac{T_s\bs{I}(t_k)^T \bs{W}(t_k)}{G(\bs{W}(t_k))}&\leq \bkt[\bigg]{1+\frac{T_s}{G(\bs{W}(t_k))}}\bkt[\Bigg]{\bkt[\Big]{T_k \bs{I}(t_k)^T\bkt[\big]{\bm{\rho}^{-1}-\bs{1}}} - \bs{I}(t_k)^T\Delta \widetilde{\bs{W}}(t_{k+1})}\\
&\hspace{60pt}+ T_k\bs{I}(t_{k+1})^T\bs{1}+\bs{I}(t_{k+1})^T\Delta \widetilde{\bs{W}}(t_{k+1})\\
&\leq T_k \bkt[\bigg]{{\frac{\mu_{\max}(T_s+1)}{\lambda_{\min}}}+1}+(T_s+2)\sum_{i=1}^{N}{\lvert\Delta\widetilde{\bs{W}}_i(t_{k+1})\rvert}.
\end{align}
By the Functional Law of Iterated Logarithm \cite{Chen2001}, with probability one we have
\begin{equation}
\Delta\widetilde{{W}}_{i}(t_{k+1}) = \textrm{O}(\sqrt{T_k \log \log T_k}), \hspace{12pt} \forall i=1,...,N.
\end{equation}
Therefore, we have
\begin{equation}
\bkt[\bigg]{\sum_{k=1}^{K_{\T}}{\bkt[\Big]{\bs{1}^T \bs{W}(t_k)}^{1-\alpha}}}\leq \frac{N}{T_s}{\sum_{k=1}^{K_{\T}}\bkt[\bigg]{{\bkt[\Big]{{\frac{\mu_{\max}(T_s+1)}{\lambda_{\min}}}+1}}T_k + \textrm{O}(\sqrt{T_k \log \log T_k})}}.\label{equation:HOL single-beam queue length bound stochastic}
\end{equation}
By dividing both sides of (\ref{equation:HOL single-beam queue length bound stochastic}) by $K_{\T}$ and using Theorem \ref{theorem:W-BMW single-channel Tk}, there must exist some constant $B_0<\infty$ such that
\begin{equation}
\lim_{{\T}\rightarrow \infty}\frac{\sum_{k=1}^{K_{\T}}{\bkt[\Big]{\bs{1}^T \bs{W}(t_k)}^{1-\alpha}}}{K_{\T}}\leq\frac{B_0}{\epsilon^{*}}.
\label{equation:HOL single-beam waiting time bound stochastic}
\end{equation}
By the Functional Law of Iterated Logarithm, as $\epsilon^{*}$ approaches 0, with probability one we further have
\begin{equation}
Q_i(t_k) = \lambda_i W_i(t_k) + O(\sqrt{W_i(t_k) \log \log W_i(t_k)}), \hspace{6pt} \forall i\in \mathcal{N}.
\label{equation:HOL Qi and Wi FLIL}
\end{equation}
Therefore, from (\ref{equation:HOL single-beam waiting time bound stochastic}) and (\ref{equation:HOL Qi and Wi FLIL}), there exists another constant $B_1<\infty$ such that
\begin{equation}
\lim_{{\T}\rightarrow \infty}\frac{\sum_{k=1}^{K_{\T}}{\bkt[\Big]{\bs{1}^T \bs{Q}(t_k)}^{1-\alpha}}}{K_{\T}}\leq \frac{B_1}{\epsilon^{*}}.
\end{equation}
For any $\alpha\in (0,1)$, by using Theorem \ref{theorem: W-BMW throughput-optimal}, we know that the Markov chain induced by $\{\bs{Q}(t)\}_{t\geq 0}$ is positive recurrent and hence
\begin{equation}
\label{equation: limiting average to expectation of waiting time}
\E[\Big]{\bkt[\big]{{\bs{1}^T \bs{Q}(t_k)}}^{1-\alpha}}= \lim_{\T\rightarrow \infty}\frac{\sum_{k=1}^{K_{\T}} \bkt[\big]{{\bs{1}^T \bs{Q}(t_k)}}^{1-\alpha}}{K_{\T}}
\end{equation}
By Lemma \ref{lemma: queue length bound for all t}, we obtain the queue length bound as
\begin{equation}
\lim_{\epsilon^{*}\downarrow 0}\epsilon^{*}\E[\Big]{(\bs{1}^{T}\bs{Q}(t))^{1-\alpha}}\leq B,
\end{equation}
for some finite constant $B>0$. By choosing the parameter $\alpha$ to be arbitrarily close to 0, the W-BMW scheduling policy can achieve asymptotically tight queue length upper bound and hence is delay-optimal. $\qed$
\end{proof}

\section{Simulation}
\label{section: simulation}
In this section, we explore the delay performance of the two types of BMW policies and the state-of-the-art VFMW policy through extensive simulation of the following three applications: polling systems, directional-antenna systems, and traffic control for signalized intersections. Throughout this section, the arrival and service process of each queue $i$ is Bernoulli with mean $\lambda_i$ and $\mu_i$, respectively. 

\subsection{Polling Systems With Arbitrary Service Order}
We consider a polling system of 4 parallel queues where the service order can be determined dynamically. We first check the delay performance of the BMW policies with different $\alpha$. Figure \ref{figure:Q-BMW power topology 1} and \ref{figure:Q-BMW power topology 2} show the total average queue length under Q-BMW policy in the two scenarios described in Section \ref{subsection: VFMW}. We state the scenarios here again for easy reference. 
\begin{itemize}
\item Scenario \RN{1}: $\bm{\lambda}=(0.119, 0.119, 0.119, 0.119)$, $\bm{\mu}=(0.5, 0.5, 0.5, 0.5)$
\item Scenario \RN{2}: $\bm{\lambda}=(0.08, 0.25, 0.09, 0.01)$, $\bm{\mu}=(0.8, 0.5, 0.3, 0.2)$
\end{itemize}
As stated in Theorem \ref{theorem:single-channel delay-optimal stochastic}, the Q-BMW policy achieves the smallest average delay when $\alpha$ is arbitrarily close to 0. Similarly, Figure \ref{figure:W-BMW power topology 1} and \ref{figure:W-BMW power topology 2} show that under the W-BMW policy the average delay is the smallest when $\alpha$ is arbitrarily close to 0. For consistency of simulation results of different scenarios, for the rest of the simulation we choose $\alpha=0.001$ for both the Q-BMW and W-BMW policies.

\begin{figure}[!htbp]
\centering
\subfigure[Scenario \RN{1}.]{
\includegraphics[scale=0.4]{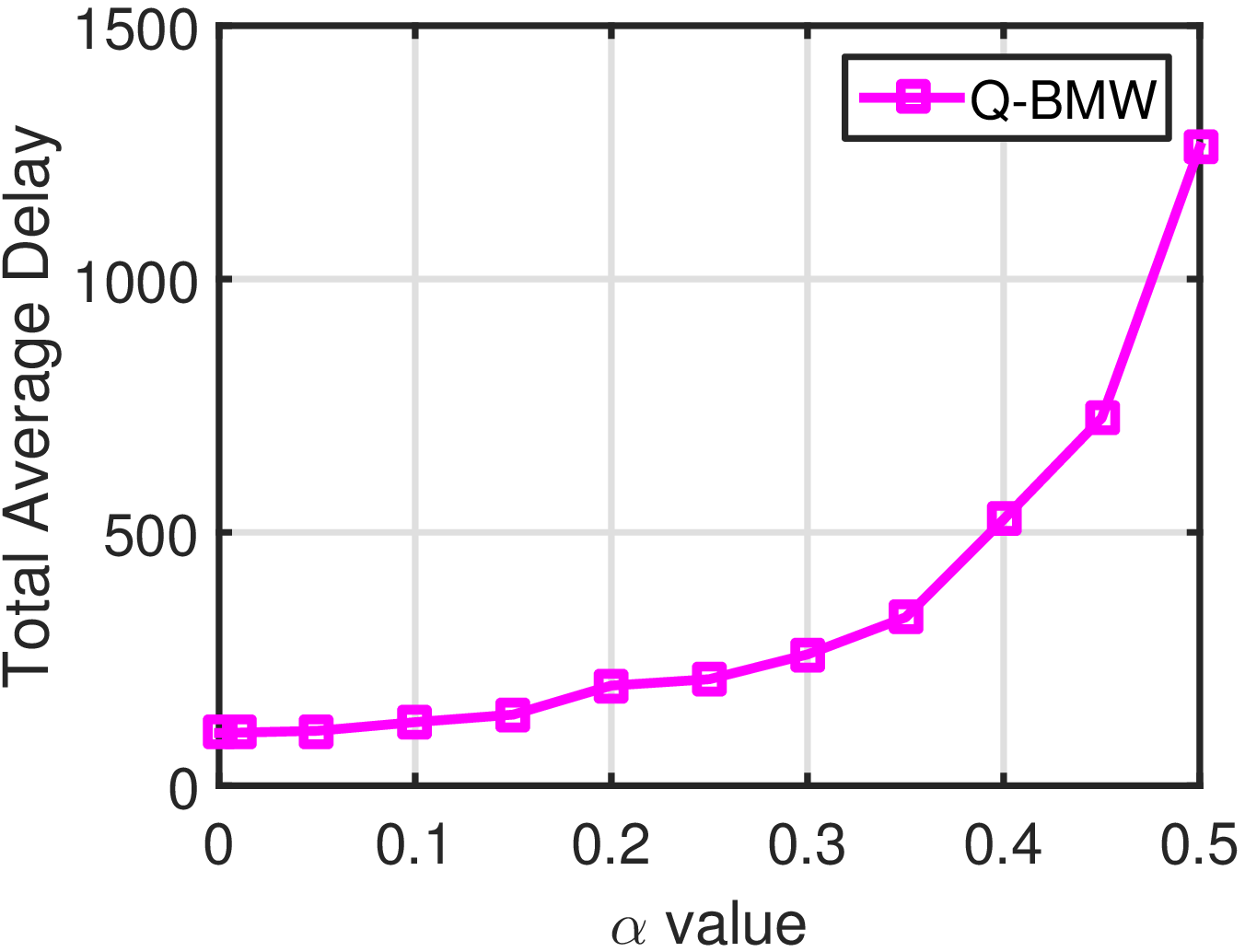}
\label{figure:Q-BMW power topology 1}}
\hspace{2mm}
\subfigure[Scenario \RN{2}.]{
\includegraphics[scale=0.4]{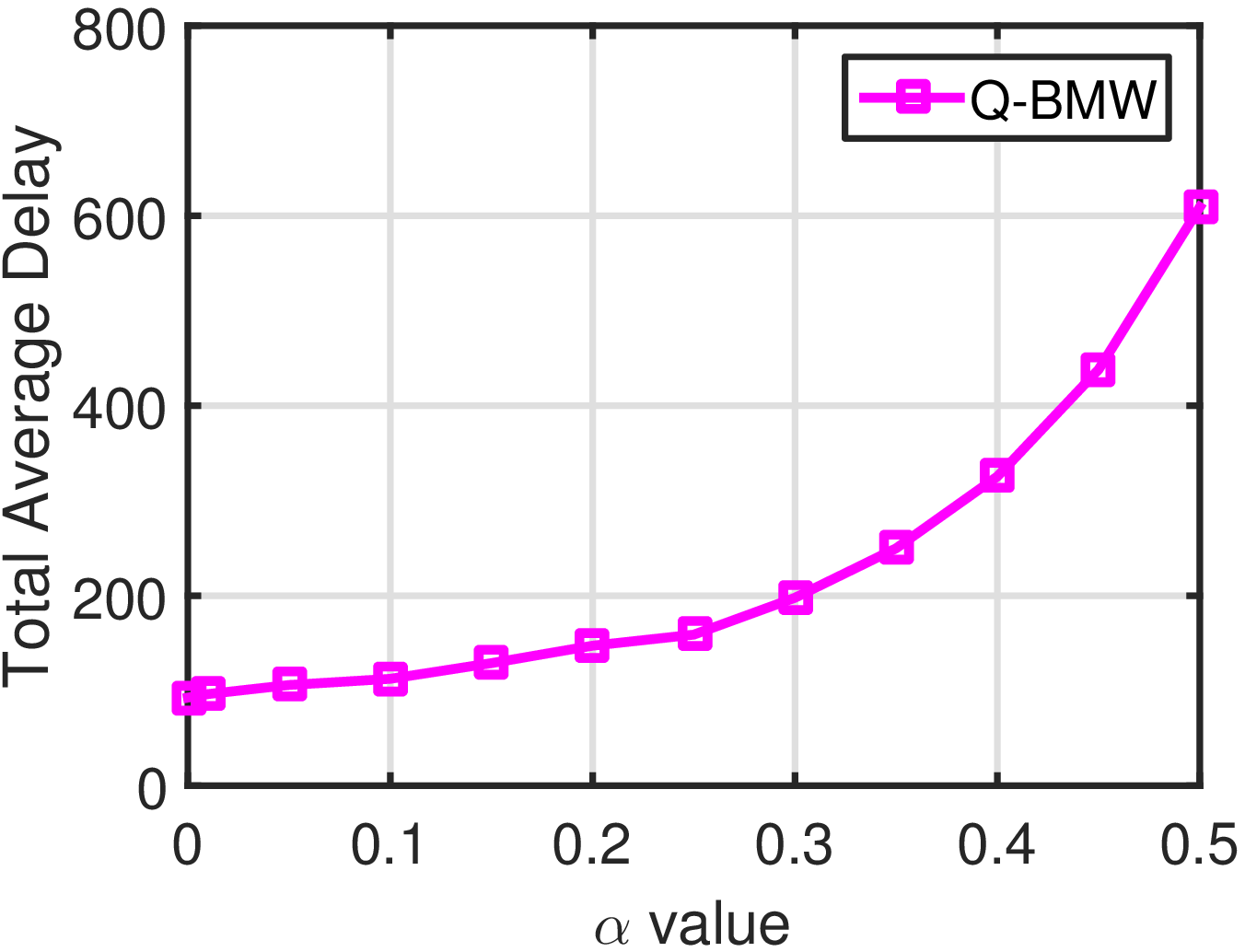}
\label{figure:Q-BMW power topology 2}}
\caption{Average delay versus different $\alpha$ value under Q-BMW policy in Scenario \RN{1} and \RN{2}.}
\end{figure}

\begin{figure}[!hbtp]
\centering
\subfigure[Scenario \RN{1}.]{
\includegraphics[scale=0.4]{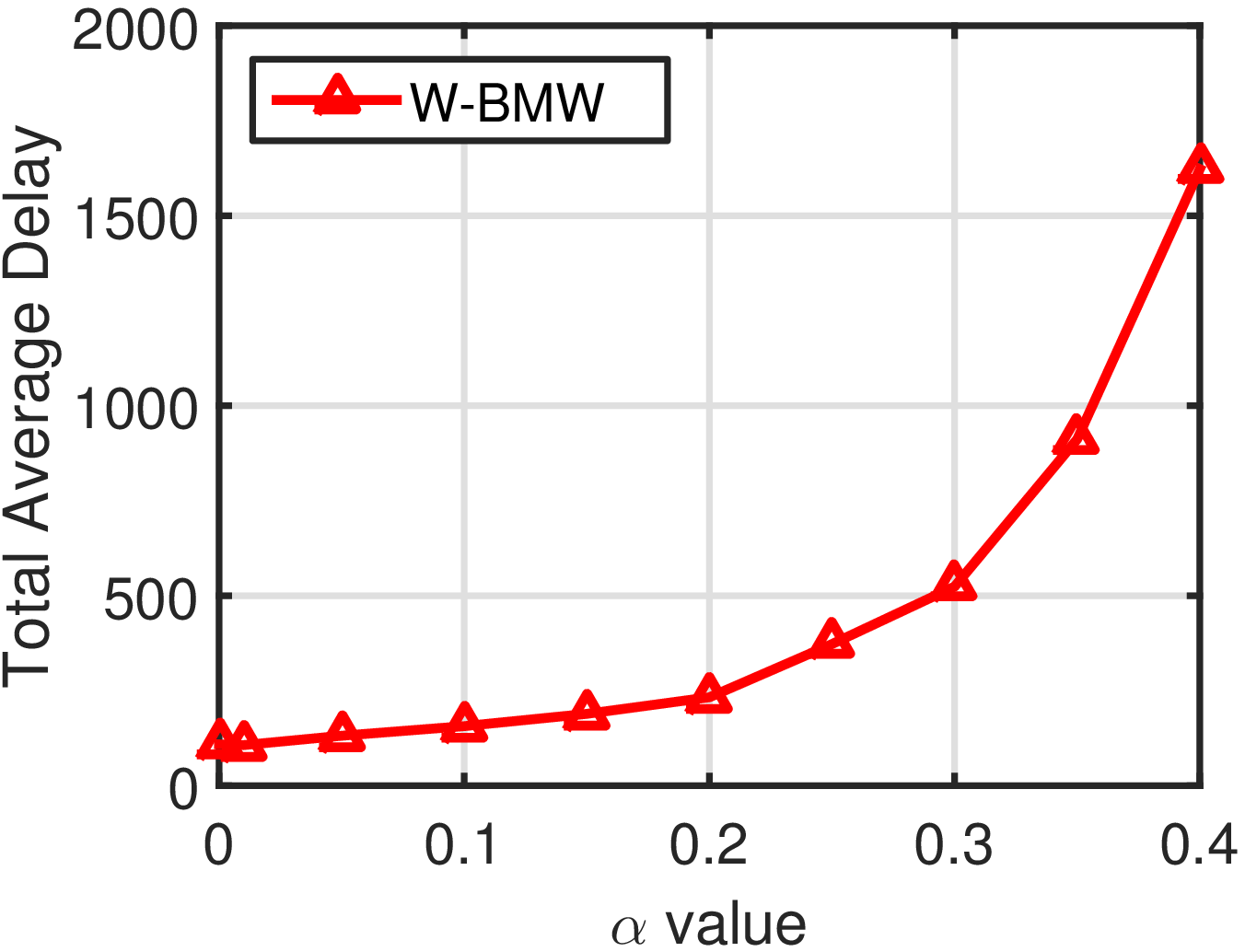}
\label{figure:W-BMW power topology 1}}
\hspace{2mm}
\subfigure[Scenario \RN{2}.]{
\includegraphics[scale=0.4]{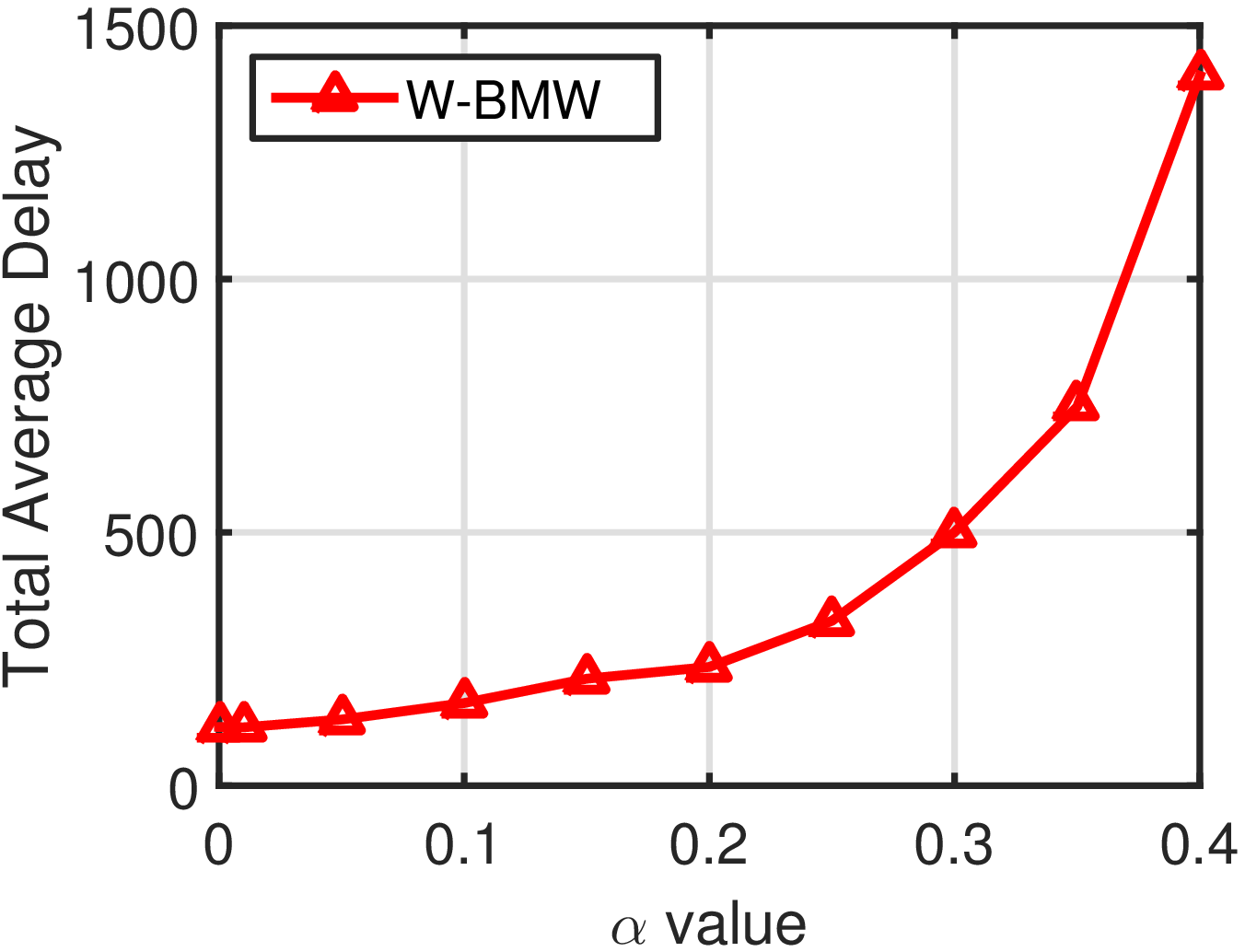}
\label{figure:W-BMW power topology 2}}
\caption{Average delay versus different $\alpha$ value under W-BMW policy in Scenario \RN{1} and \RN{2}.}
\end{figure}

Next, we simulate the average delay with different utilization factor $\beta^{*}$ under the three scheduling policies, as shown in Figure \ref{figure:1-beam rho symmetric} and \ref{figure:1-beam rho asymmetric}. We consider both the symmetric and asymmetric cases:
\begin{itemize}
\item Scenario \RN{3}: $\bm{\lambda}=\beta^{*}\cdot(0.125, 0.125, 0.125, 0.125)$, $\bm{\mu}=(0.5, 0.5, 0.5, 0.5)$
\item Scenario \RN{4}: $\bm{\lambda}=\beta^{*}\cdot(0.25, 0.15, 0.075, 0.025)$, $\bm{\mu}=(0.5, 0.5, 0.5, 0.5)$
\end{itemize}
In Figure \ref{figure:1-beam rho symmetric} and \ref{figure:1-beam rho asymmetric}, we observe that both Q-BMW and W-BMW achieve a much lower average delay than that of the VFMW policy with either frame size function $\bkt[\big]{\sum_{i}Q_i(t)}^{0.5}$ or  $\bkt[\big]{\sum_{i}Q_i(t)}^{0.99}$. Moreover, we are also interested in the delay performance with different amount of switching overhead. Figure \ref{figure:1-beam Ts symmetric} and \ref{figure:1-beam Ts asymmetric} show the total average delay with $T_s$ ranging from 1 to 7 under Scenario \RN{3} and \RN{4} with utilization factor equal to 0.95. In these two figures we do not show the simulation results for VFMW with $\alpha=0.5$ simply because its delay is much larger than its counterparts. We observe that the total average delay grows roughly linearly with the switching overhead and the two BMW policies still have much smaller delay than the VFMW policy, regardless of the amount of switching overhead. Therefore, for the rest of the simulation, we simply choose $T_s=1$. Figure \ref{figure:1-beam rho asymmetric foreach} shows the per-queue average delay of the two BMW policies in Scenario \RN{4}. Under the Q-BMW policy, the per-queue delay is inversely proportional to the mean arrival rate. On the other hand, the delay of each queue is about the same under the W-BMW policy. Hence, W-BMW indeed achieves better fairness in per-queue delay than the Q-BMW policy.

\begin{figure}[!hbtp]
\centering
\subfigure[Scenario \RN{3}.]{
\includegraphics[scale=0.4]{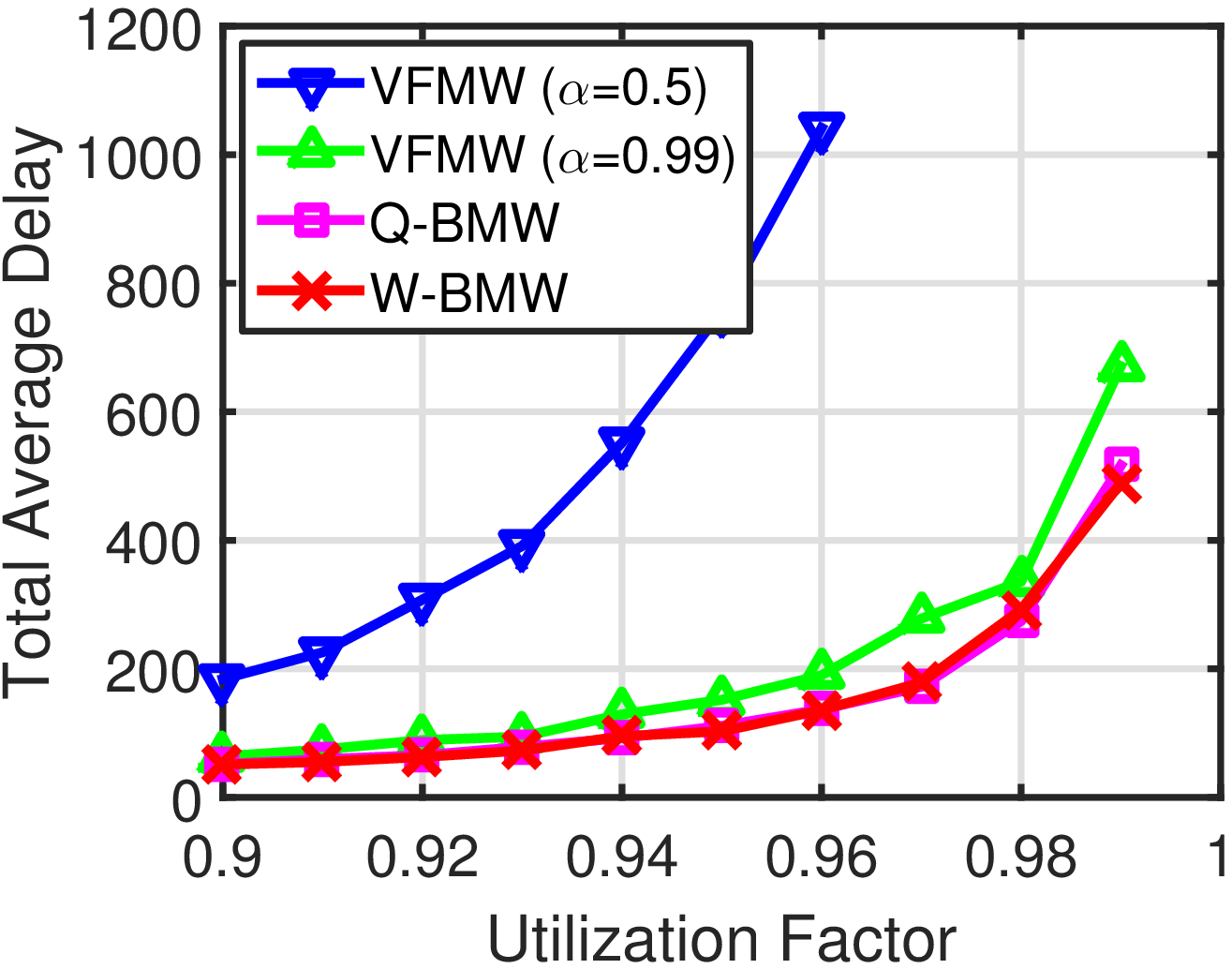}
\label{figure:1-beam rho symmetric}}
\hspace{2mm}
\subfigure[Scenario \RN{4}.]{
\includegraphics[scale=0.4]{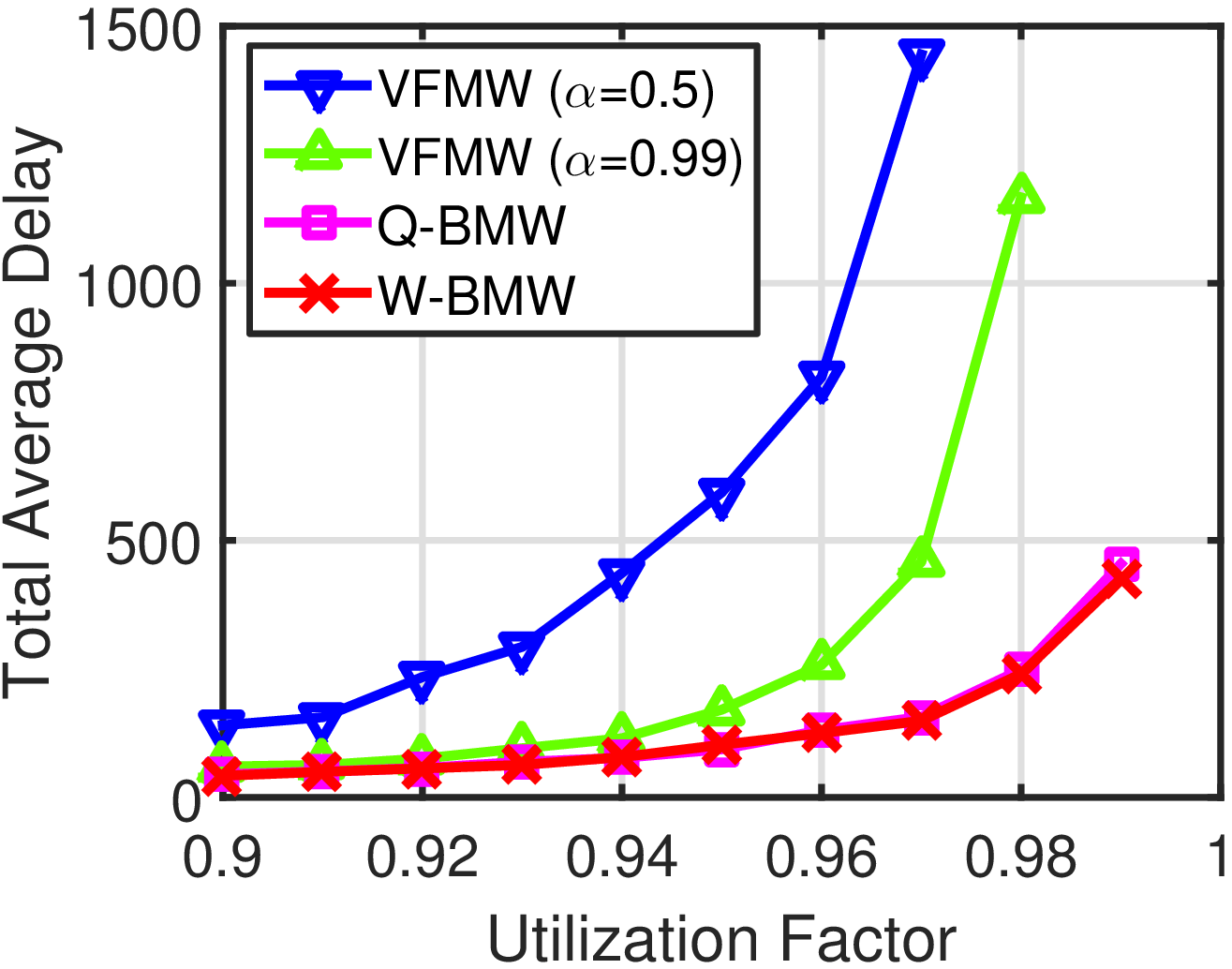}
\label{figure:1-beam rho asymmetric}}
\caption{Total average delay under the Q-BMW, W-BMW, and VFMW policies in Scenario \RN{3} and \RN{4}.}
\end{figure}

\begin{figure}[!hbtp]
\centering
\subfigure[Scenario \RN{3}.]{
\includegraphics[scale=0.4]{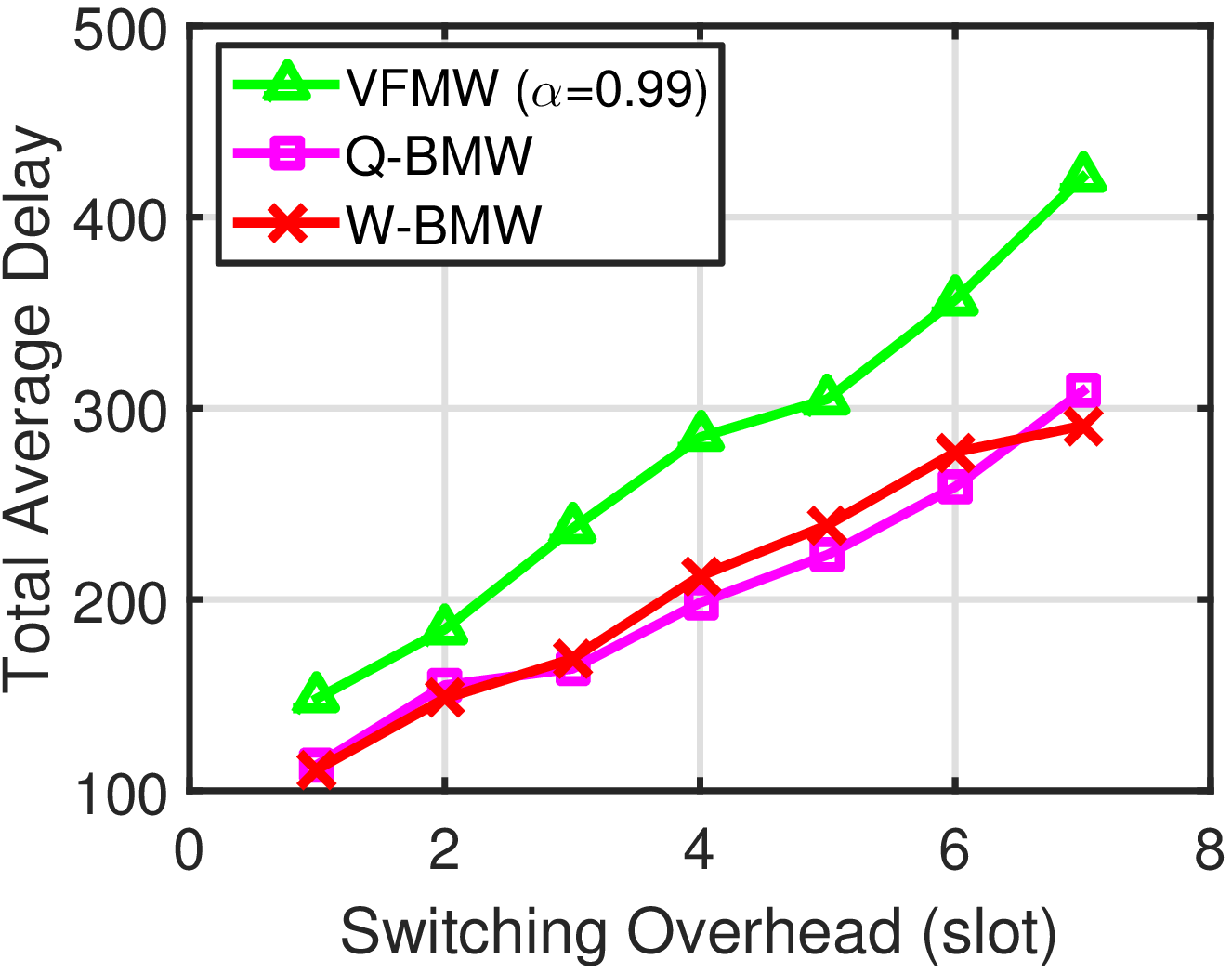}
\label{figure:1-beam Ts symmetric}}
\hspace{2mm}
\subfigure[Scenario \RN{4}.]{
\includegraphics[scale=0.4]{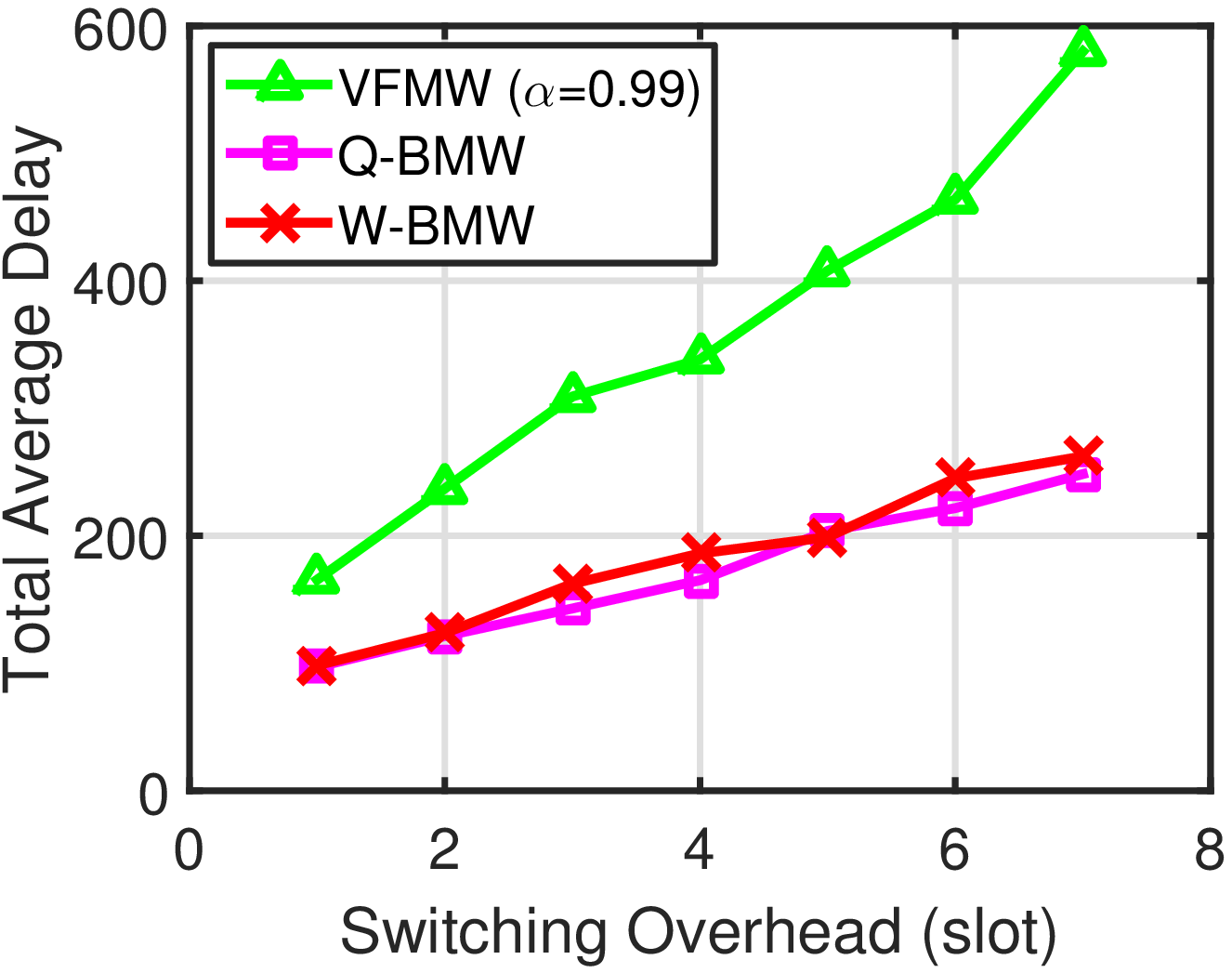}
\label{figure:1-beam Ts asymmetric}}
\caption{Total average delay versus different amount of switching overhead in Scenario \RN{3} and \RN{4}.}
\end{figure}

\begin{figure}[!htbp]
\centering
\includegraphics[scale=0.4]{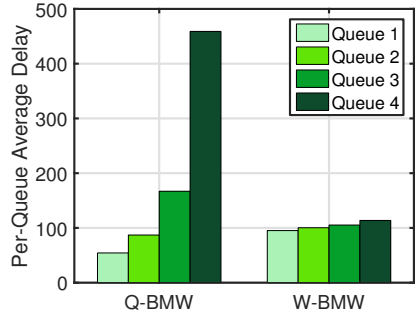}
\caption{Per-queue delay under the Q-BMW and W-BMW policies in Scenario \RN{4}.}
\label{figure:1-beam rho asymmetric foreach}
\end{figure}

\subsection{Multi-Beam Directional-Antenna Systems}
We consider a system of 6 queues and $\lvert \bs{I}\rvert=4$ for every feasible schedule $\bs{I}$. In the context of directional-antenna systems, this example represents a 4-beam system. 
Besides, there are system-wise conflicting constraints which limit the number of maximal feasible schedules. We consider the topology with sets of conflicting queues: queue 1 and queue 2 cannot be served simultaneously; queue 3 and queue 4 cannot be served simultaneously.
Therefore, there are only four maximal feasible schedules: $\{1, 3, 5, 6\}, \{1, 4, 5, 6\}, \{2, 3, 5, 6\},$ and $\{2, 4, 5, 6\}$. Besides, we choose the traffic pattern to be :
\begin{itemize}
\item Scenario \RN{5}: $\bm{\lambda}=\beta^{*}\cdot(0.18, 0.16, 0.25, 0.3, 0.9, 0.8)$ and asymmetric service rates $\bm{\mu}=(0.3, 0.4, 0.5, 0.6, 0.9, 0.8)$,
\end{itemize}
with different utilization factor $\beta^{*}$. First, we measure the average delay under the two BMW policies with different $\alpha$. Figure \ref{figure:4-beam rho power Q-BMW} and \ref{figure:4-beam rho power W-BMW} show that the average delay gets lower with smaller $\alpha$ for both Q-BMW and W-BMW policy. This result is consistent with that of the queueing systems where the server can serve at most one queue at a time. Next, with $\alpha$ equal to 0.001, Figure \ref{figure:4-beam rho interference 1} shows that the two BMW policies still achieve much smaller delay than that of the VFMW policy.
\begin{figure}[!htbp]
\begin{center}
\subfigure[Q-BMW scheduling.]{
\includegraphics[scale=0.4]{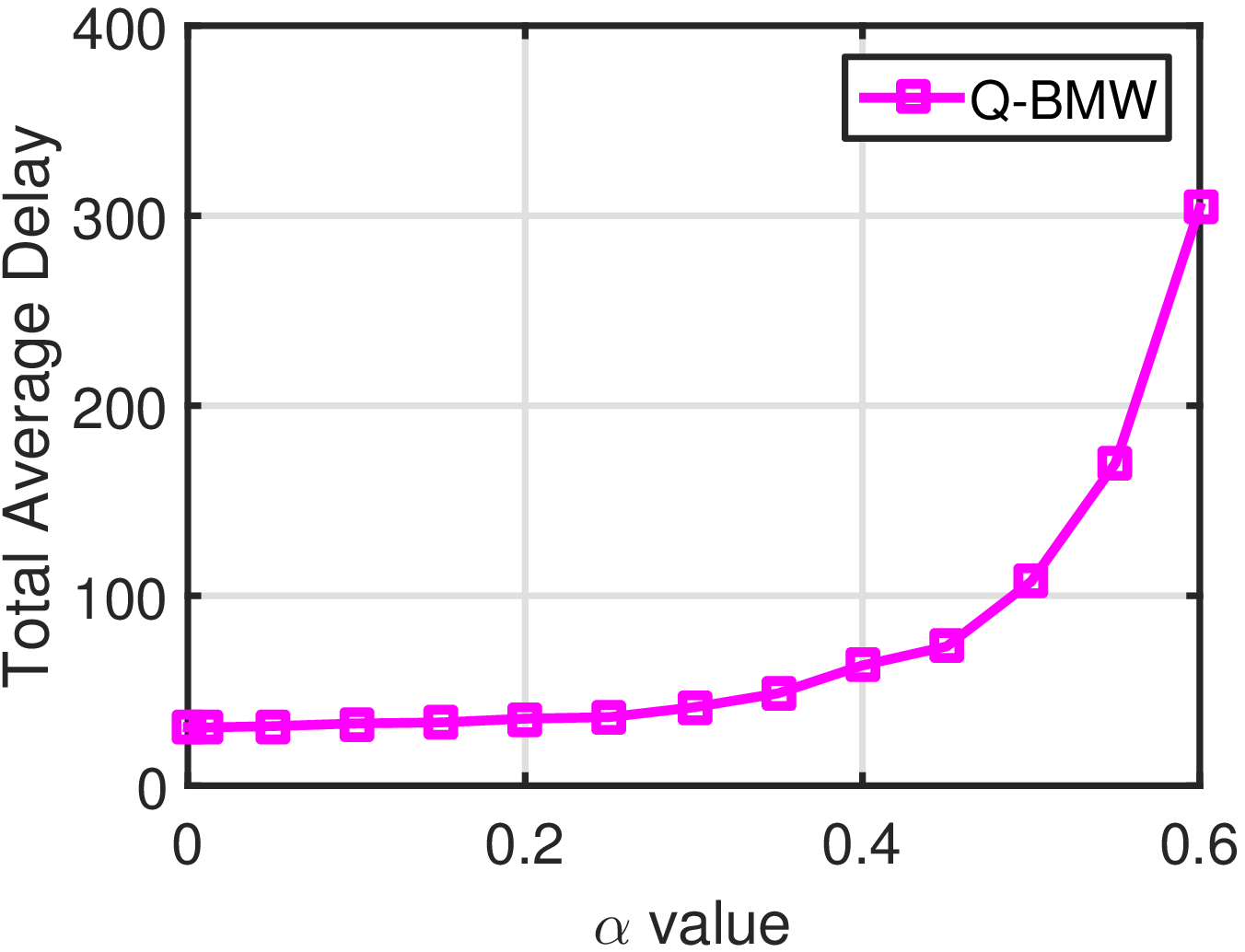}
\label{figure:4-beam rho power Q-BMW}}
\hspace{2mm}
\subfigure[W-BMW scheduling]{
\includegraphics[scale=0.4]{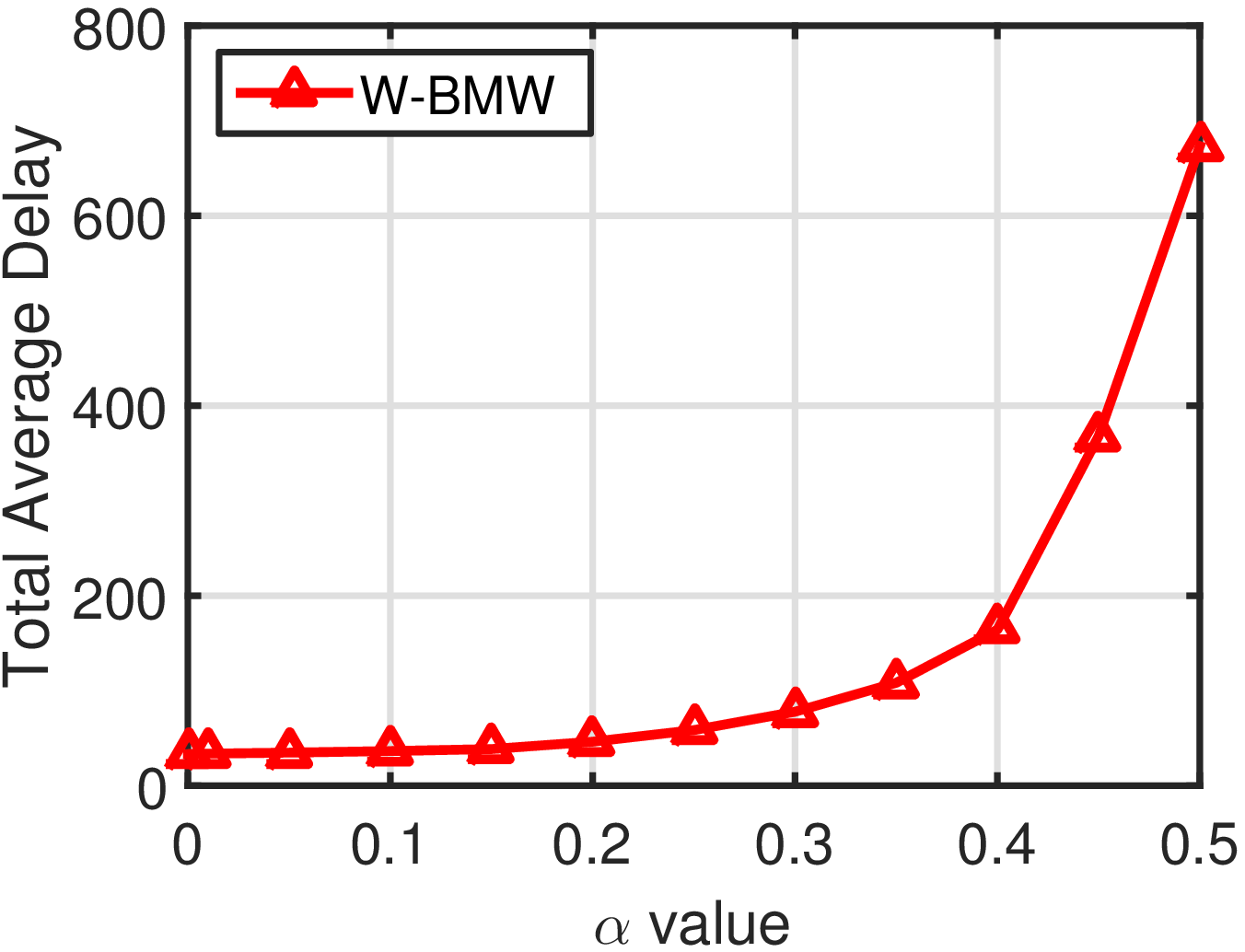}
\label{figure:4-beam rho power W-BMW}}
\caption{Average delay under Q-BMW and W-BMW policies with different $\alpha$ in Scenario \RN{5}.}
\end{center}
\end{figure}
\begin{figure}[!htbp]
\begin{center}
\subfigure[Scenario \RN{5}.]{
\includegraphics[scale=0.4]{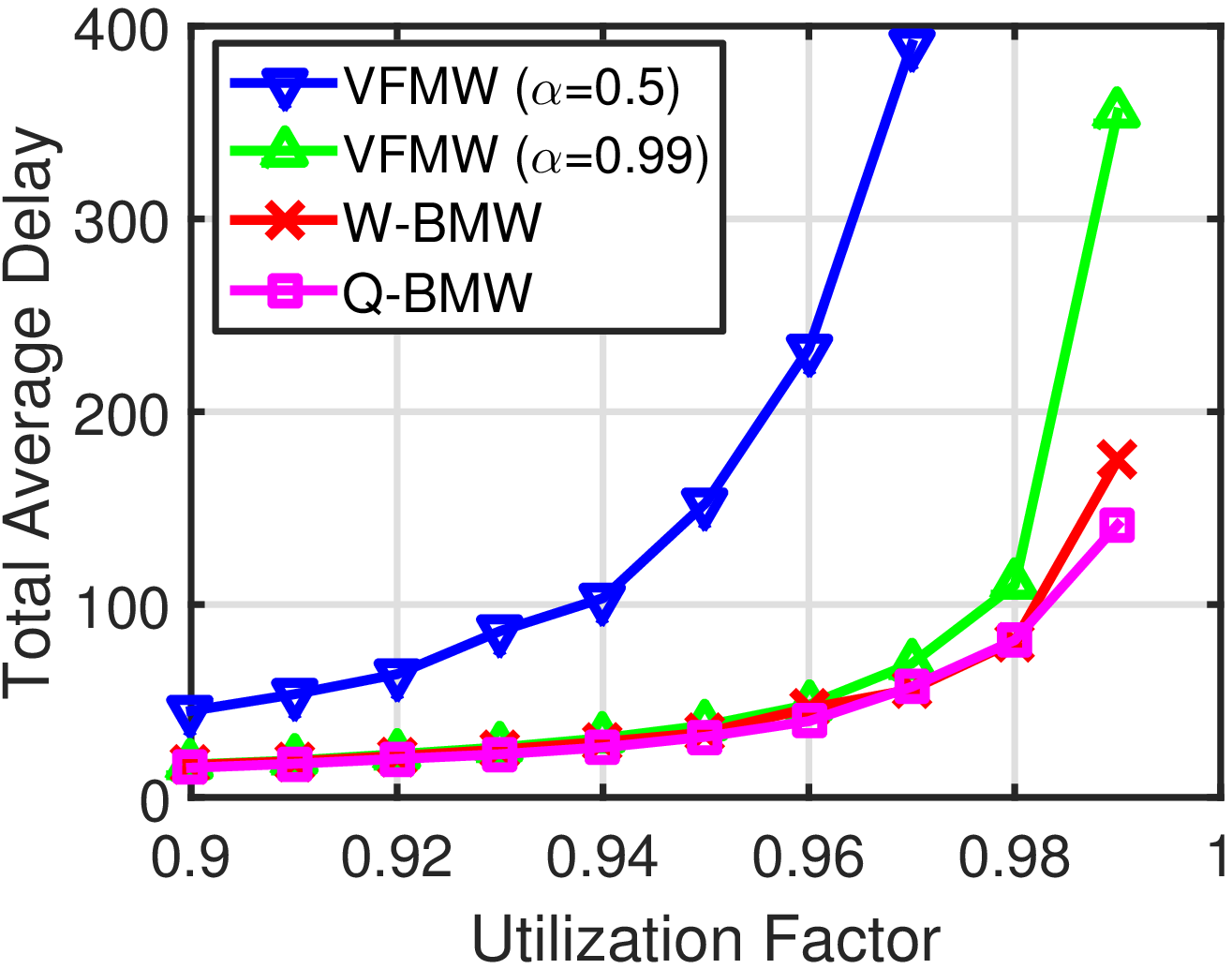}
\label{figure:4-beam rho interference 1}}
\hspace{2mm}
\subfigure[Scenario \RN{6}.]{
\includegraphics[scale=0.4]{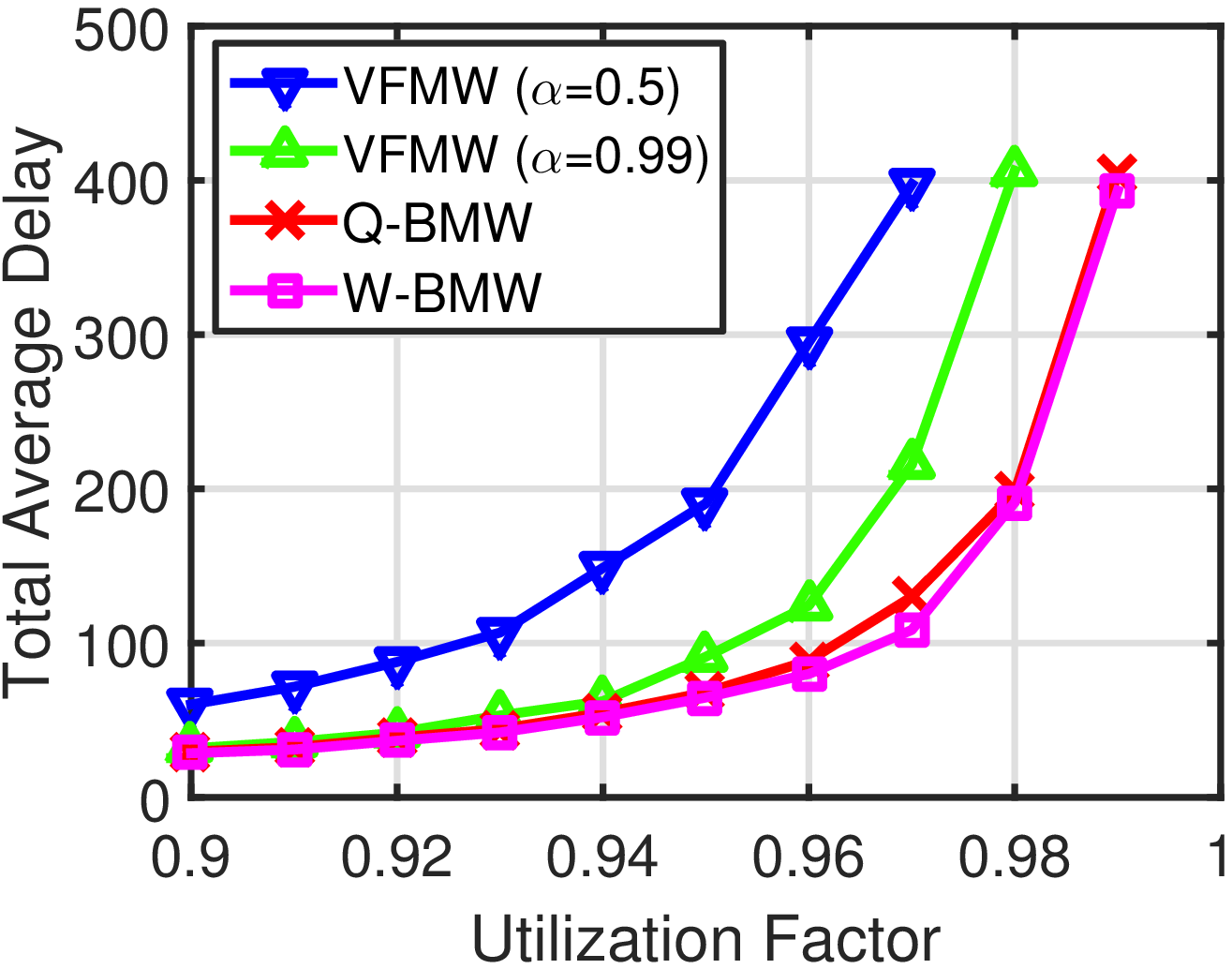}
\label{figure:4-beam rho interference 2}}
\caption{Delay comparison of the Q-BMW, W-BMW, and VFMW policies in Scenario \RN{5} and \RN{6}.}
\end{center}
\end{figure}
Using the same topology as in Figure \ref{figure:4-beam rho interference 1}, we change the arrival rates and service rates to be:
\begin{itemize}
\item Scenario \RN{6}: $\bm{\lambda}=\beta^{*}\cdot(0.35, 0.15, 0.3, 0.2, 0.5, 0.5)$ and ${\mu_i}=0.5$ for every $i$.
\end{itemize} 
As shown in Figure \ref{figure:4-beam rho interference 2}, the result is consistent with that of the other traffic pattern. In summary, with totally different traffic patterns, both Q-BMW and W-BMW always achieve much smaller average delay than that of the VFMW scheduling policy, regardless of the traffic pattern.

\subsection{Isolated Signalized Intersections}
We consider an isolated four-way signalized intersection at which each arriving vehicle either goes straight or makes a left turn. The intersection can be modeled by a queueing system with 8 queues (4 through lanes and 4 left-turn lanes) and $\lvert \bs{I}\rvert=2$ for every feasible schedule $\bs{I}$. Moreover, due to conflicting constraints imposed by the system, there are six maximal feasible schedules. We consider two different traffic patterns as follows:
 \begin{itemize}
\item Scenario \RN{7}: $\bm{\lambda}=\beta^{*}\cdot(0.1, 0.5, 0.1, 0.3, 0.1, 0.5, 0.1, 0.3)$ and ${\mu_i}=1$ for every queue $i$.
\item Scenario \RN{8}: $\bm{\lambda}=\beta^{*}\cdot(0.02, 0.26, 0.24, 0.48, 0.24, 0.48, 0.02, 0.26)$ and ${\mu_i}=1$ for every queue $i$.
\end{itemize}
Note that the service processes are chosen to be deterministic since the amount of vehicles that are able to pass through the intersection in one time slot should exhibit very little variation. Figure \ref{figure:intersection rho topo 1} and \ref{figure:intersection rho topo 2} show the average delay under the three policies in Scenario \RN{7} and \RN{8}. Note that for the VFMW policy we choose $\alpha=0.8$ instead of $\alpha=0.99$ simply because the average delay with $\alpha=0.99$ turns out to be extremely large in these two scenarios. Again, the two BMW policies achieve better system-wise delay performance than the VFMW policy. Besides, note that Figure \ref{figure:intersection rho topo 1} shows that VFMW with $\alpha=0.5$ performs better than VFMW with $\alpha=0.8$, while Figure \ref{figure:intersection rho topo 2} shows the opposite. This also highlights the fundamental dilemma of choosing $\alpha$ for the VFMW policy, as discussed in Section \ref{subsection: VFMW}. We further compare the per-queue average delay under Q-BMW and W-BMW. Figure \ref{figure:intersection rho topo 1 foreach} and \ref{figure:intersection rho topo 2 foreach} show the per-queue delay of queue 1 through queue 4. For simplicity, we do not show the results for queue 5 to queue 8 since they have the same arrival rate pattern and hence have similar per-delay performance as queue 1 to queue 4. These two figures demonstrate that W-BMW still achieves much better fairness than Q-BMW in the sense that the queues with lighter traffic do not suffer from huge queueing delay. Since the per-queue delay is especially crucial in transportation systems, the W-BMW policy is particularly suitable for traffic control at signalized intersections. 

\begin{figure}[!hbtp]
\centering
\subfigure[Scenario \RN{7}.]{
\includegraphics[scale=0.4]{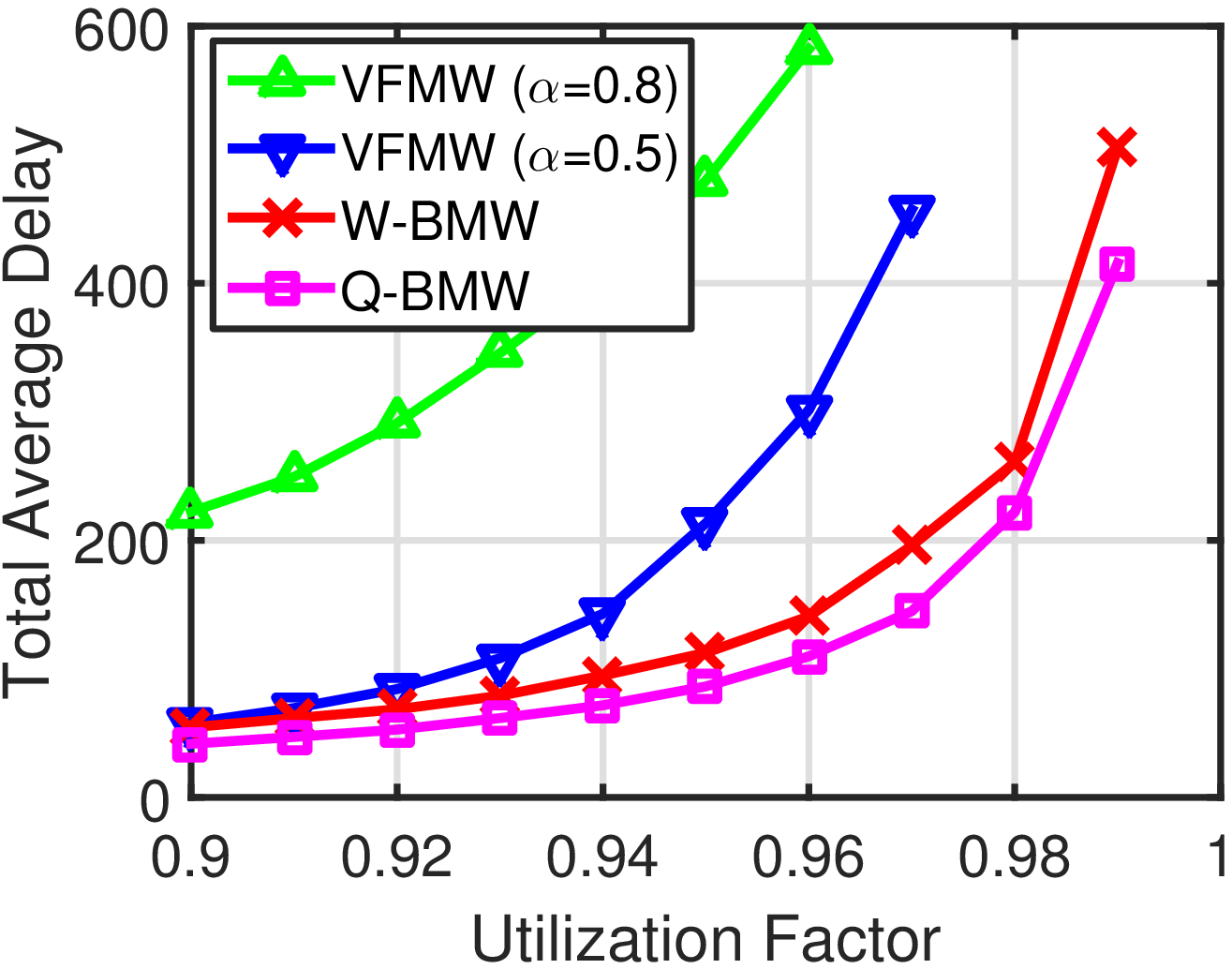}
\label{figure:intersection rho topo 1}}
\hspace{2mm}
\subfigure[Scenario \RN{8}.]{
\includegraphics[scale=0.4]{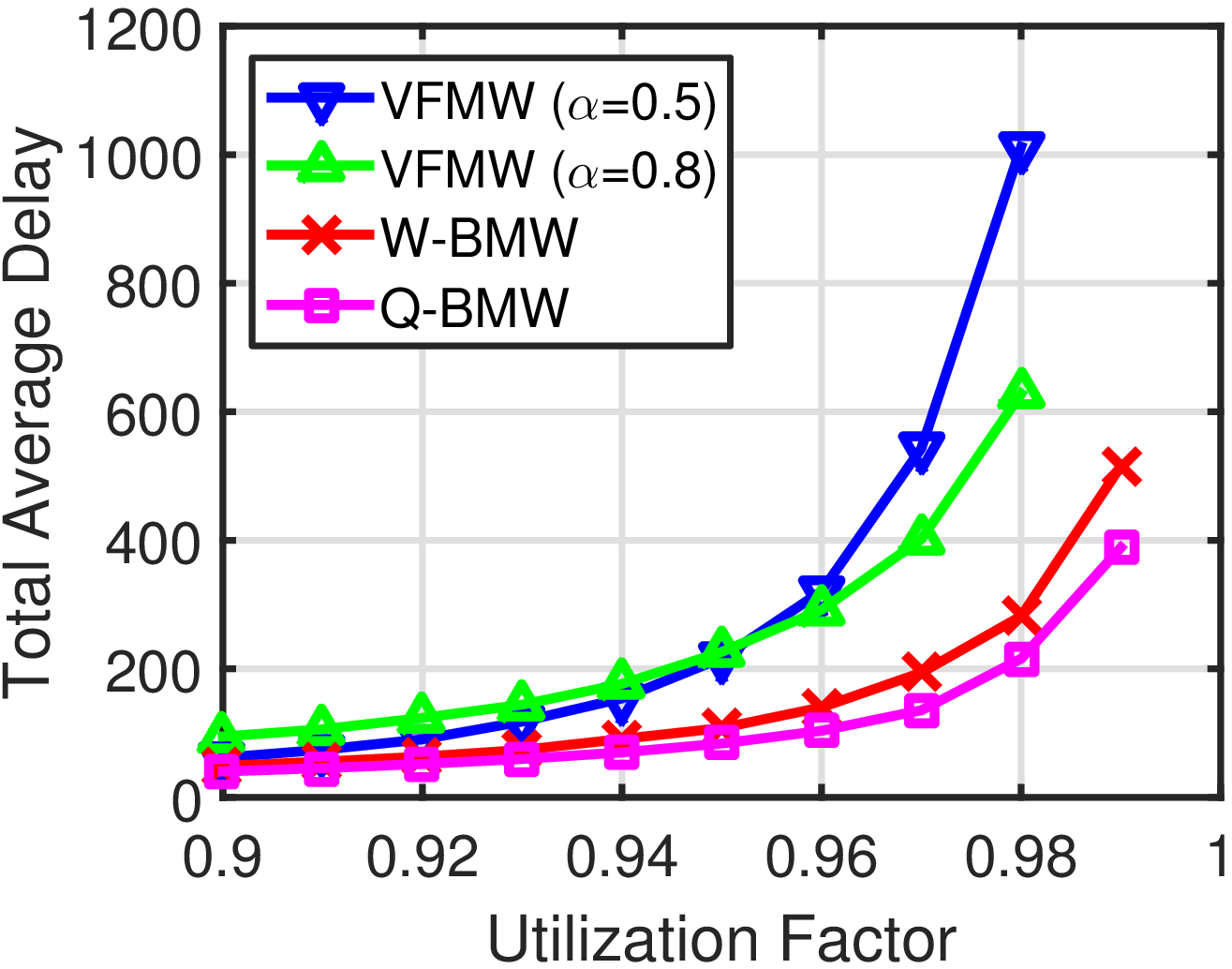}
\label{figure:intersection rho topo 2}}
\caption{Total average delay under the Q-BMW, W-BMW, and VFMW policies in Scenario \RN{7} and \RN{8}.}
\end{figure}

\begin{figure}[!hbtp]
\centering
\subfigure[Scenario \RN{7}.]{
\includegraphics[scale=0.4]{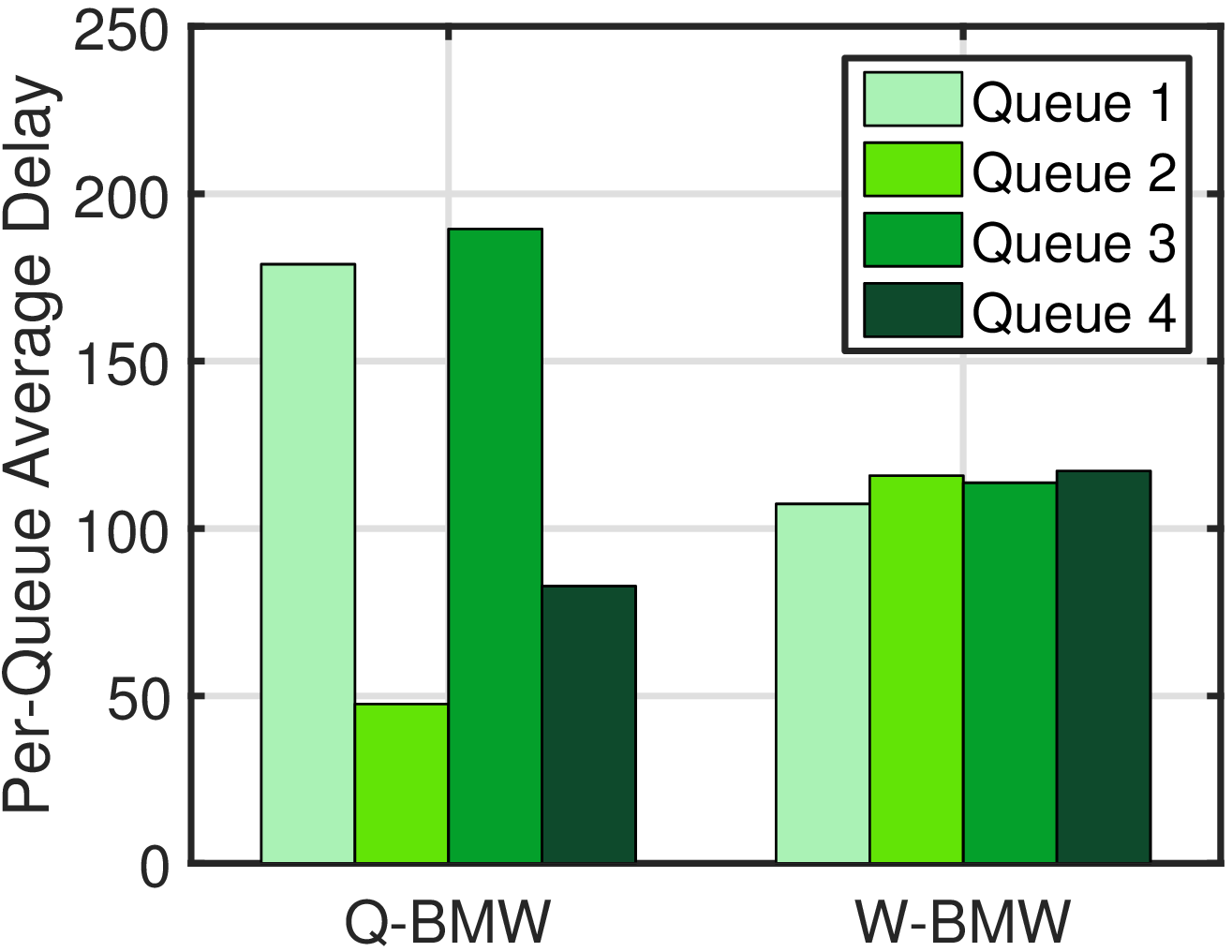}
\label{figure:intersection rho topo 1 foreach}}
\hspace{2mm}
\subfigure[Scenario \RN{8}.]{
\includegraphics[scale=0.4]{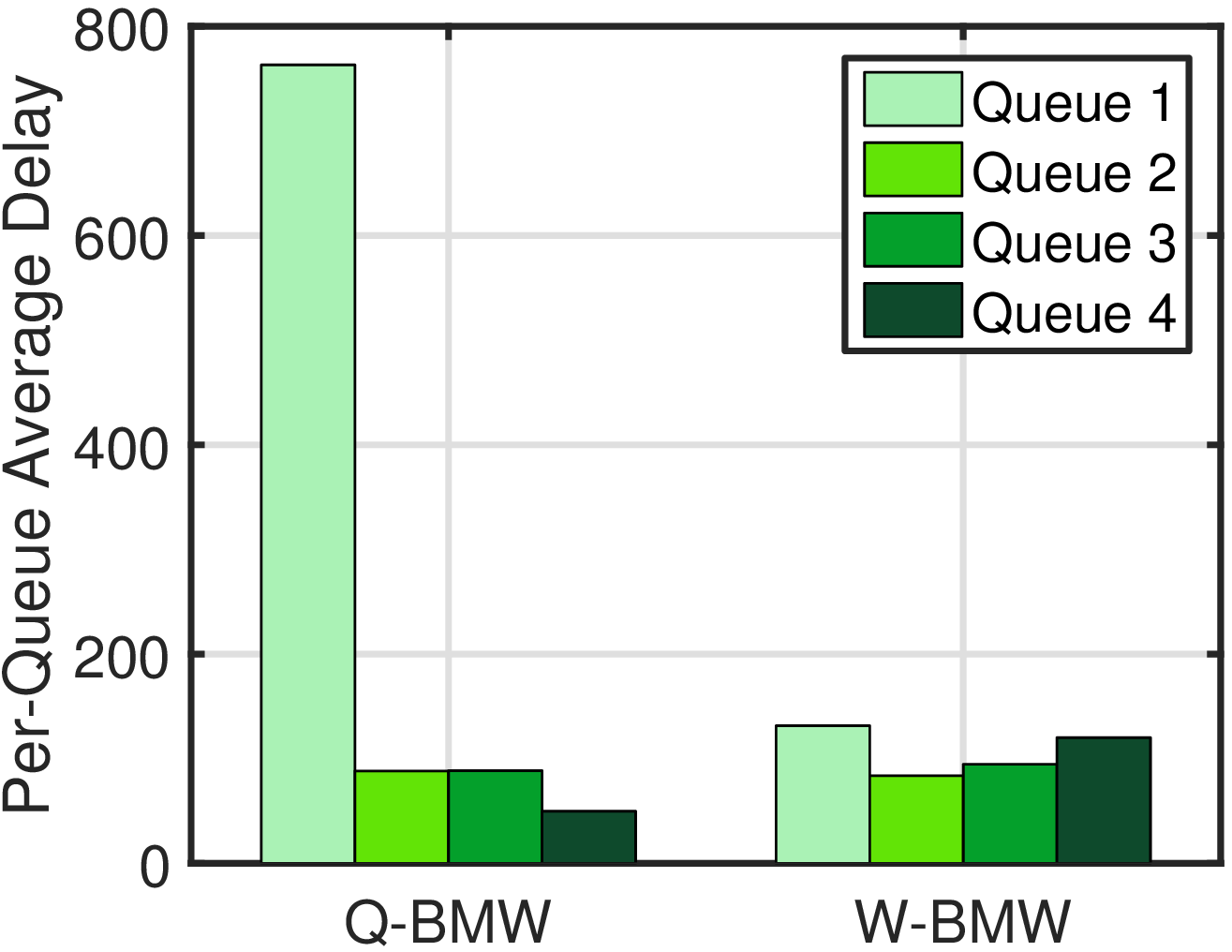}
\label{figure:intersection rho topo 2 foreach}}
\caption{Per-queue average delay under Q-BMW and W-BMW in Scenario \RN{7} and \RN{8}.}
\end{figure}

\section{Conclusion}
\label{section: conclusion}
In this paper, we study the delay performance of queueing systems with switching overhead. We propose two types of BMW scheduling policies that achieves not only throughput-optimality but also delay-optimality. We provide a theoretical queue length upper bound which is asymptotically tight. Through extensive simulation, we demonstrate that the proposed policies achieve much better delay performance than that of the state-of-the-art policy.

\begin{acknowledgements}
This material is based upon work supported in part by the U. S. Army Research Laboratory and the U. S. Army Research Office under contract/grant number W911NF-15-1-0279 and NPRP Grant 8-1531-2-651 of Qatar National Research Fund (a member of Qatar Foundation).
\end{acknowledgements}

\section*{Appendix 1: Proof of Theorem \ref{theorem: Q-BMW throughput-optimal}}
\label{appendix: Q-BMW}
We prove the theorem by choosing a proper Lyapunov function and showing that the Lyapunov drift is negative. To begin with, consider the multi-step queue length evolution: In the $k$-th interval, for any $\tau$ with $0<\tau\leq T_k$, for any queue $i$ we have
\begin{equation}
Q_i(t_k+\tau)\leq \max\left\{0, Q_i(t_k)-\sum_{s=0}^{\tau-1}M(t_k+s)I_{i}(t_k+s)S_i(t_k+s)\right\}+\sum_{s=0}^{\tau-1}A_i(t_k+s).
\end{equation}
Therefore, we also have
\begin{align}
Q_i(t_k+\tau)^2&\leq \bkt[\bigg]{Q_i(t_k)-\sum_{s=0}^{\tau-1}M(t_k+s)I_{i}(t_k+s)S_i(t_k+s)}^{2}+\bkt[\bigg]{\sum_{s=0}^{\tau-1}A_i(t_k+s)}^{2}\\ 
&\hspace{12pt}+2\sum_{s=0}^{\tau-1}A_i(t_k+s) \bkt[\bigg]{Q_i(t_k)-\sum_{s=0}^{\tau-1}M(t_k+s)I_{i}(t_k+s)S_i(t_k+s)}\\
&\leq Q_i(t_k)^2 + 2Q_i(t_k)\bkt[\bigg]{\sum_{s=0}^{\tau-1}A_i(t_k+s)-\sum_{s=0}^{\tau-1}M(t_k+s)I_{i}(t_k+s)S_i(t_k+s)}\label{equation:Qi evolution 1}\\
&\hspace{12pt}+\bkt[\bigg]{\sum_{s=0}^{\tau-1}M(t_k+s)I_{i}(t_k+s)S_i(t_k+s)}^2+\bkt[\bigg]{\sum_{s=0}^{\tau-1}A_i(t_k+s)}^{2}. \label{equation:Qi evolution 2}
\end{align}
Define a Lyapunov function 
\begin{equation}
L(t):=\bs{Q}(t)^{T}\bs{U}\bs{Q}(t), 
\end{equation}
where $\bs{U}:=\textrm{diag}(\mu_1^{-1},\cdots, \mu_N^{-1})$. Define $\widetilde{F}(\bs{Q}(t))=(\bs{1}^T \bs{Q}(t))^{\alpha_1}$ with $\alpha_1\in(0,\alpha)$ and $\alpha+\alpha_1<1$. Let $\widetilde{T}_k=\min\{T_k, \widetilde{F}(\bs{Q}(t_k))\}$. For any $\tau_1, \tau_2 \geq 0$, define the conditional drift between $\tau_1$ and $\tau_2$ as $\Delta(\tau_1, \tau_2):=\E[\big]{L(\tau_2)-L(\tau_1) \given \bs{Q}(\tau_1)}$.
The conditional drift between $t_k$ and $t_k+\widetilde{T}_k$ is
\begin{align}
\Delta(t_k&, t_k+\widetilde{T}_k)\\
= & \E[\Big]{\bs{Q}(t_k+\widetilde{T}_k)^T \bs{U} \bs{Q}(t_k+\widetilde{T}_k)-\bs{Q}(t_k)^T \bs{U} \bs{Q}(t_k) \given \bs{Q}(t_k)} \\
\leq & 2\cdot \E[\Bigg]{\bkt[\bigg]{\sum_{t=t_{k}}^{t_k+\widetilde{T}_k-1}(\bs{A}(t)-\bs{S}^{*}(t))^{T}}\bs{U}\bs{Q}(t_k) \given \bs{Q}(t_k)}\label{equation:ASUQ} \\ 
+& \E[\Bigg]{\bkt[\bigg]{\sum_{t=t_k}^{t_k+\widetilde{T}_k-1}\bs{A}(t)}^T \bs{U}\bkt[\bigg]{\sum_{t=t_k}^{t_k+\widetilde{T}_k-1}\bs{A}(t)} \given \bs{Q}(t_k)}\label{equation:ATUA} \\
+& \E[\Bigg]{\bkt[\bigg]{\sum_{t=t_k}^{t_k+\widetilde{T}_k-1}\bs{S}^{*}(t)}^T \bs{U}\bkt[\bigg]{\sum_{t=t_k}^{t_k+\widetilde{T}_k-1}\bs{S}^{*}(t)} \given \bs{Q}(t_k)},\label{equation:STUS}
\end{align}
where $\bs{S}^{*}(t):=M(t)(\bs{S(t)}\circ \bs{I}(t))$ and $\bs{S(t)}\circ \bs{I}(t)$ denotes the element-wise product (also called Hadamard product) of the two vectors $\bs{S}(t)$ and $\bs{I}(t)$. For any $t\geq 0$, we have $\bs{A}(t)\leq A_{\max}\cdot\bs{1}$ and $\bs{S}(t)\leq S_{\max}\cdot\bs{1}$, regardless of the queue length at time $t$. Therefore, (\ref{equation:ATUA}) and (\ref{equation:STUS}) are bounded as
\begin{align}
&\E[\Bigg]{\bkt[\bigg]{\sum_{t=t_k}^{t_k+\widetilde{T}_k-1}\bs{A}(t)}^T \bs{U}\bkt[\bigg]{\sum_{t=t_k}^{t_k+\widetilde{T}_k-1}\bs{A}(t)} \given \bs{Q}(t_k)}\leq A_{\max}^2 \Tr(\bs{U}) \E[\big]{{\widetilde{T}_{k}}^2\given \bs{Q}(t_k)} \label{equation:AUA bounded}\\
&\E[\Bigg]{\bkt[\bigg]{\sum_{t=t_k}^{t_k+\widetilde{T}_k-1}\bs{S}^{*}(t)}^T \bs{U}\bkt[\bigg]{\sum_{t=t_k}^{t_k+\widetilde{T}_k-1}\bs{S}^{*}(t)} \given \bs{Q}(t_k)}\leq  S_{\max}^2 \Tr(\bs{U}) \E[\big]{{\widetilde{T}_{k}}^2\given \bs{Q}(t_k)}. \label{equation:SUS bounded}
\end{align}
Since both $\bs{A}(t)$ and $\bs{S}(t)$ are independent of $\bs{Q}(t_k)$, we can rewrite (\ref{equation:ASUQ}) as
\begin{align}
& \E[\Bigg]{\bkt[\bigg]{\sum_{t=t_{k}}^{t_k+\widetilde{T}_k-1}(\bs{A}(t)-\bs{S}^{*}(t))^{T}}\bs{U}\bs{Q}(t_k) \given \bs{Q}(t_k)}\\
&= \E[\Big]{\bkt[\big]{\widetilde{T}_k \bm{\rho}^T-(\widetilde{T}_k-T_{s})\bs{I}(t_k)^T}\bs{Q}(t_k)\given \bs{Q}(t_k)},
\end{align}
where $\bm{\rho}$ is the vector of normalized traffic load of each queue. Since $\bm{\lambda}$ is assumed to be in the capacity region, then there exists a $J$-dimensional non-negative vector $\bm{\beta}^T=\bkt{\beta_1,\cdots,\beta_J}$ with $\bm{\beta}^T \bs{1}<1$ such that $\sum_{j=1}^{J}\beta_j \bs{I}^{(j)}\geq \bm{\rho}$. Under the Q-BMW policy, it is guaranteed that $\bs{I}(t_k)^T \bs{Q}(t_k)=\max_{j:1\leq j\leq J} (\bs{I}^{(j)})^T \bs{Q}(t_k)$. Therefore, 
\begin{align}
\bm{\rho}^T\bs{Q}(t_k)&\leq \bkt[\bigg]{\sum_{j=1}^{J}\beta_j\bs{I}^{(j)}}^T \bs{Q}(t_k)\label{equation:rho to epsilon 1}\\
&\leq \bkt[\bigg]{\sum_{j=1}^{J}\beta_j}\bs{I}(t_k)^T \bs{Q}(t_k)\label{equation:rho to epsilon 2} \\
&= (1-\epsilon)\bs{I}(t_k)^T \bs{Q}(t_k)\label{equation:rho to epsilon 3}
\end{align}
where $\epsilon:=1-\bm{\beta}^T\bs{1}$ denotes the corresponding "distance" from the boundary of the capacity region.
From (\ref{equation:ATUA})-(\ref{equation:rho to epsilon 3}), the conditional drift can be written as
\begin{align}
\Delta(t_k&, t_k+\widetilde{T}_k)\leq \E[\Big]{2\cdot\bkt[\big]{-\epsilon\widetilde{T}_k + T_s}\bs{I}(t_k)^T \bs{Q}(t_k)+B_0\widetilde{T}_k^{2}\given \bs{Q}(t_k)},\label{equation: drift bound 1}
\end{align}
where $B_0=\bkt[\big]{S_{\max}^2+A_{\max}^2}\Tr\bkt[\big]{\bs{U}}$ does not depend on the queue length vector or scheduling decisions. Suppose the server can serve at most $K$ queues at a time. By Lemma \ref{lemma:lower bound for Tk}, we also know that
\begin{equation}
T_k\geq{C_0\bkt[\Big]{{\bs{1}}^T\bs{Q}(t_k)}^{1-\alpha}},
\end{equation}
where $C_0=T_s/\bkt[\big]{NK\bkt[\big]{A_{\max}+(1+T_s)S_{\max}}}$. Here, we need to discuss two possible cases:

\vspace{3mm}
\noindent {\bf Case 1}:  $\widetilde{F}(\bs{Q}(t_k))\geq {C_0\bkt[\big]{{\bs{1}}^T\bs{Q}(t_k)}^{1-\alpha}}$

\vspace{1mm}
\noindent The above condition also implies that
\begin{align}
& \bkt[\Big]{{{\bs{1}}^T\bs{Q}(t_k)}}^{\alpha_1}\geq {C_0\bkt[\Big]{{\bs{1}}^T\bs{Q}(t_k)}^{1-\alpha}}.
\end{align}
Therefore, by the assumption that $\alpha+\alpha_1 < 1$, we have
\begin{equation}
{\bs{1}}^T\bs{Q}(t_k)\leq C_0^{\frac{-1}{1-\alpha-\alpha_1}}.\label{equation:case 1 queue bound}
\end{equation}
Hence, from (\ref{equation: drift bound 1}) and (\ref{equation:case 1 queue bound}), we know that the conditional drift between $t_k$ and $t_k+\widetilde{T}_k$ is bounded, i.e.
\begin{align}
\Delta(t_k, t_k+\widetilde{T}_k)&\leq \E[\Big]{2\cdot\bs{I}(t_k)^T \bs{Q}(t_k)+B_0\widetilde{T}_k^{2}\given \bs{Q}(t_k)} \\
& \leq  2T_s\cdot\bs{I}(t_k)^T \bs{Q}(t_k)+ {B_0\widetilde{F}(\bs{Q}(t_k))^2} \\
& \leq  2T_s\cdot\bs{1}^T \bs{Q}(t_k)+B_0\bkt[\Big]{{{\bs{1}}^T\bs{Q}(t_k)}}^{2\alpha_1}\\
& \leq  2T_s\cdot{C_0}^{\frac{-1}{1-\alpha-\alpha_1}}+B_0  {C_0}^{\frac{-2\alpha_1}{1-\alpha-\alpha_1}}< \infty.
\end{align} 
This also implies that the unconditional drift between $t_k$ and $t_k+\widetilde{T}_k$ is bounded, i.e.
\begin{equation}
\E[\Big]{L(t_k+\widetilde{T}_k)-L(t_k)} \leq 2T_s\cdot{C_0}^{\frac{-1}{1-\alpha-\alpha_1}}+B_0 {C_0}^{\frac{-2\alpha_1}{1-\alpha-\alpha_1}}< \infty.
\end{equation}
\noindent {\bf Case 2}:  $\widetilde{F}(\bs{Q}(t_k))< {C_0\bkt[\big]{{\bs{1}}^T\bs{Q}(t_k)}^{1-\alpha}}$

\vspace{1mm}
\noindent The above condition also implies that $\widetilde{T}_k=\widetilde{F}(\bs{Q}(t_k))<T_k$. Therefore, (\ref{equation: drift bound 1}) can then be written as
\begin{align}
\Delta(t_k, t_k+\widetilde{T}_k)&\leq 2\bkt[\Big]{-\epsilon\widetilde{F}(\bs{Q}(t_k)) + T_s}\bs{I}(t_k)^T \bs{Q}(t_k)+ {B_0\widetilde{F}(\bs{Q}(t_k))^2} \\
& \leq -\frac{2\epsilon}{N} \bkt[\Big]{{{\bs{1}}^T\bs{Q}(t_k)}}^{1+\alpha_1}+2T_s\bkt[\Big]{{{\bs{1}}^T\bs{Q}(t_k)}}+B_0 \bkt[\Big]{{{\bs{1}}^T\bs{Q}(t_k)}}^{2\alpha_1}.
\end{align}
Since $-\frac{2\epsilon}{N} \bkt[\big]{{{\bs{1}}^T\bs{Q}(t_k)}}^{1+\alpha_1}$ is the dominating term, there must exist some constant $B_2>0$ such that
\begin{equation}
\Delta(t_k, t_k+\widetilde{T}_k) \leq B_2 - \frac{\epsilon}{N}\bkt[\Big]{{{\bs{1}}^T\bs{Q}(t_k)}}^{1+\alpha_1}.\label{equation:case2 conditional drift}
\end{equation}
Moreover, we also know that
\begin{align}
\sum_{t=t_k}^{t_k-\widetilde{T}_k -1} \bs{1}^T \bs{Q}(t)\leq \sum_{t=t_k}^{t_k-\widetilde{T}_k -1}\bs{1}^T \bkt[\bigg]{\bs{Q}(t_k)+\sum_{\tau=t_k}^{t_k+\widetilde{T}_k-1} \bs{A}(\tau)}
\end{align}
By taking conditional expectation, we have
\begin{align}
\E[\Bigg]{\sum_{t=t_k}^{t_k-\widetilde{T}_k -1} \bs{1}^T \bs{Q}(t)\sgiven \bs{Q}(t_k)}&\leq \E[\bigg]{\widetilde{T}_k\cdot \bs{1}^T \bs{Q}(t_k)+{\widetilde{T}_k}^2 \cdot \bs{1}^T \bm{\lambda}\sgiven \bs{Q}(t_k)}\\
&= \bkt[\Big]{\bs{1}^T\bs{Q}(t_k)}^{1+\alpha_1}+ \bs{1}^T \bm{\lambda}\cdot\bkt[\Big]{\bs{1}^T\bs{Q}(t_k)}^{2\alpha_1}\\
&\leq (1+\bs{1}^T \bm{\lambda})\bkt[\Big]{\bs{1}^T\bs{Q}(t_k)}^{1+\alpha_1}.\label{equation:bound of sum queue length}
\end{align}
The last inequality holds since $\alpha_1<1$. Therefore, based on (\ref{equation:case2 conditional drift}) and (\ref{equation:bound of sum queue length}), we obtain
\begin{equation}
\E[\Big]{L(t_k+\widetilde{T}_k)-L(t_k)} \leq B_2 - \frac{\epsilon}{N_1}\E[\Bigg]{\sum_{t=t_k}^{t_k+\widetilde{T}_k-1}{\bs{1}^T\bs{Q}(t)}},\label{equation:case2 unconditional drift bound}
\end{equation}
where $N_1:=N(1+\bs{1}^T \bm{\lambda})$.

Next, we consider the slot-by-slot conditional drift for any $t$ between $t_k+\widetilde{T}_k$ and $t_{k+1}$. Note that there is no switching between $t_k+\widetilde{T}_k$ and $t_{k+1}$ and hence $M(t)=1$ for all $t\in [t_k+\widetilde{T}_k,t_{k+1})$. Therefore,
\begin{align}
\Delta(t,t+1)&= \E[\Big]{\bs{Q}(t+1)^T \bs{U} \bs{Q}(t+1)-\bs{Q}(t)^T \bs{U} \bs{Q}(t) \given \bs{Q}(t)}\\
&\leq 2\cdot \E[\Big]{\bkt[\big]{\bs{A}(t)-\bs{S}^{*}(t)}^{T}\bs{U}\bs{Q}(t) \given \bs{Q}(t)}\label{equation:one-slot ASUQ} \\ 
&\hspace{12pt} + \E[\Big]{{\bs{A}(t)}^T \bs{U}{\bs{A}(t)} + {\bs{S}^{*}(t)}^T \bs{U}{\bs{S}^{*}(t)}\given \bs{Q}(t)}\label{equation:one-slot STUS}
\end{align}
Similar to (\ref{equation:AUA bounded}) and (\ref{equation:SUS bounded}), we know 
\begin{equation}
\E[\Big]{{\bs{A}(t)}^T \bs{U}{\bs{A}(t)} + {\bs{S}^{*}(t)}^T \bs{U}{\bs{S}^{*}(t)}}\leq \bkt[\big]{A_{\max}^2+S_{\max}^2}\Tr({\bs{U}})=B_0.
\end{equation}
Besides, since $\bs{A}(t)$ and $\bs{S}(t)$ are independent of $\bs{Q}(t)$, (\ref{equation:one-slot ASUQ}) can be written as
\begin{align}
& \E[\Big]{\bkt[\big]{\bs{A}(t)-\bs{S}^{*}(t)}^{T}\bs{U}\bs{Q}(t) \given \bs{Q}(t)}= {\bkt[\big]{\bm{\rho}^T-\bs{I}(t)^T}\bs{Q}(t)}
\end{align}
Hence, we have
\begin{equation}
\Delta(t,t+1) \leq 2\cdot {\bkt[\big]{\bm{\rho}^T-\bs{I}(t)^T}\bs{Q}(t)}+B_0\label{equation:one-step drift}
\end{equation}
Under the Q-BMW policy, at time $t$ we must have
\begin{align}
\bs{I}(t)^T \bs{Q}(t)&\geq \bkt[\big]{\bs{I}^{(j)}}^T \bs{Q}(t)-\frac{\bs{I}(t)^T\bs{Q}(t)}{F(\bs{Q}(t_k))}, \hspace{12pt}\forall j=1,...,N
\end{align}
Along with (\ref{equation:rho to epsilon 1})-(\ref{equation:rho to epsilon 3}), we then have
\begin{align}
(1-\epsilon)\bs{I}(t)^T \bs{Q}(t)&=\sum_{j=1}^{J}\beta_j \bs{I}(t)^T \bs{Q}(t)\\
& \geq \sum_{j=1}^{J} \beta_j \bkt[\big]{\bs{I}^{(j)}}^T \bs{Q}(t)-\sum_{j=1}^{J}\beta_j \frac{\bs{I}^T(t)\bs{Q}(t)}{F(\bs{Q}(t_k))}\\
& \geq \bm{\rho}^{T} \bs{Q}(t)- (1-\epsilon)\frac{\bs{I}^T(t)\bs{Q}(t)}{F(\bs{Q}(t_k))}.
\end{align}
Therefore, (\ref{equation:one-step drift}) can be written as
\begin{align}
\Delta(t,t+1)& \leq 2 \cdot \bkt[\bigg]{-\epsilon \bs{I}(t)^T\bs{Q}(t)+(1-\epsilon)\frac{\bs{I}(t)^T\bs{Q}(t)}{F(\bs{Q}(t_k))}}+B_0\label{equation:one-step drift 1}\\
&\leq 2 \cdot  \bkt[\bigg]{-\frac{\epsilon}{N} \bkt[\big]{\bs{1}(t)^T\bs{Q}(t)}+(1-\epsilon)\frac{\bs{I}(t)^T\bs{Q}(t)}{F(\bs{Q}(t_k))}}+B_0\label{equation:one-step drift 2} \\
&\leq  2 \cdot  \bkt[\bigg]{-\frac{\epsilon}{N} \bkt[\big]{\bs{1}(t)^T\bs{Q}(t)}+(1-\epsilon)\bkt[\Big]{{\bs{I}(t)^T\bs{Q}(t)}}^{1-\alpha}}+B_0\label{equation:one-step drift 3} \\
&\leq  2 \cdot  \bkt[\bigg]{-\frac{\epsilon}{N} \bkt[\big]{\bs{1}(t)^T\bs{Q}(t)}+(1-\epsilon)\bkt[\Big]{{\bs{1}(t)^T\bs{Q}(t)}}^{1-\alpha}}+B_0\label{equation:one-step drift 4}
\end{align}
Since $\alpha>0$, then $-\frac{\epsilon}{N} \bs{I}(t)^T\bs{Q}(t)$ is the dominating term in (\ref{equation:one-step drift 4}). In other words, there must exist some constant $B_3>0$ such that
\begin{align}
\Delta(t,t+1)& \leq B_3 - \frac{\epsilon}{N} \bkt[\Big]{\bs{1}(t)^T\bs{Q}(t)}.
\end{align}
Hence, for any $t\in (t_k+\widetilde{T}_{k}, t_{k+1})$, we know
\begin{equation}
\E[\Big]{L(t+1)-L(t)}\leq B_3- \frac{\epsilon}{N} \E[\Big]{{\bs{1}(t)^T\bs{Q}(t)}}.\label{equation:one-step unconditional drift bound}
\end{equation}
Now, we consider any large $\T$ and let $K_{\T}$ be the number of intervals in $[0,\T)$. Since each interval lasts for at least one slot, then $K_{\T}\leq \T$. The unconditional drift in $[0,\T)$
\begin{align}
\E[\big]{L(\T)-L(0)} &= \sum_{k=0}^{K_{\T}-1} \E[\big]{L(t_{k+1})-L(t_k)}\\
& = \sum_{k=0}^{K_{\T}-1} \bkt[\Bigg]{\E[\big]{L(t_k+\widetilde{T}_k)-L(t_k)}+\sum_{\tau=t_k+\widetilde{T}_k}^{t_{k+1}-1}\E[\big]{L(\tau+1)-L(\tau)}  }\\
& \leq K_{\T}B_2 - \frac{\epsilon}{N_1}\sum_{k=1}^{K_{\T}-1}\bkt[\Bigg]{\E[\Bigg]{\sum_{t=t_k}^{t_k+\widetilde{T}_k-1}{\bs{1}^T\bs{Q}(t)}}} \\
&\hspace{60pt}+ B_3\T -  \frac{\epsilon}{N} \sum_{k=0}^{K_{\T}-1} \bkt[\Bigg]{\E[\Bigg]{\sum_{\tau=t_k+\widetilde{T}_k}^{t_{k+1}-1} {\bs{1}^T\bs{Q}(t)}}}\\
&\leq  K_{\T}B_2 + B_3\T  - \frac{\epsilon}{N_1}\bkt[\Bigg]{\E[\Bigg]{\sum_{t=0}^{\T-1}{\bs{1}^T\bs{Q}(t)}}}.
\end{align}
Since  $L(0)=0$ and $L(t)$ is nonnegative regardless of $t$, by letting $\T\rightarrow \infty$, we have
\begin{equation}
\lim_{\T\rightarrow \infty}\frac{\E[\Big]{\sum_{t=0}^{\T-1}{\bs{1}^T\bs{Q}(t)}}}{\T}\leq \frac{N_1(B_2+B_3)}{\epsilon}< \infty. 
\end{equation}
$\qed$

\section*{Appendix 2: Proof of Theorem \ref{theorem: W-BMW throughput-optimal}}
To begin with, we describe how the HOL waiting time evolves. For each queue $i$, define $\delta_{i}(t):=I_{i}(t)\sum_{m=1}^{S_{i}(t)}V_{i}(\varphi_{i}(t)+m)$. Note that $\delta_{i}(t)$ represents the potential decrease in HOL waiting time due to the potential service of queue $i$ at time $t$. We use the boldface symbol $\bm{\delta}$ to denote the $N$-dimensional vector $(\delta_1,...,\delta_{N})\in \mathbb{N}_{0}^{N}$. In the $k$-th interval, for any $\tau$ with $0<\tau\leq T_k$, for any queue $i$ we have
\begin{equation}
W_i(t_k+\tau)\leq \tau+\max\left\{0, W_i(t_k)-\sum_{s=0}^{\tau-1}M(t_k+s)\delta_{i}(t_k+s)\right\}.
\end{equation}
Recall that $M(t)$ represents whether the server is in ACTIVE mode at time $t$.
Therefore, we also have
\begin{align}
W_i(t_k+\tau)^2&\leq W_i(t_k)^2 - 2W_i(t_k)\bkt[\bigg]{\sum_{s=0}^{\tau-1}M(t_k+s)\delta_{i}(t_k+s)}\label{equation:Wi evolution 1}\\
&+\bkt[\bigg]{\sum_{s=0}^{\tau-1}M(t_k+s)\delta_{i}(t_k+s)}^2 + 2\tau W_i(t_k)+\tau^2. \label{equation:Wi evolution 2}
\end{align}
Define a Lyapunov function 
\begin{equation}
L_2(t):=\bs{W}(t)^{T}\bs{P}\bs{W}(t), 
\end{equation}
where $\bs{P}:=\textrm{diag}(\rho_1,\cdots, \rho_N)$. Define $\hat{G}(\bs{W}(t))=(\bs{1}^T \bs{W}(t))^{\alpha_1}$ with $\alpha_1\in(0,\alpha)$ and $\alpha+\alpha_1<1$. Let $\hat{T}_k=\min\{T_k, \hat{G}(\bs{W}(t_k))\}$. For any $\tau_1, \tau_2 \geq 0$, define the conditional drift between $\tau_1$ and $\tau_2$ as $\Delta L_2(\tau_1, \tau_2):=\E[\big]{L_2(\tau_2)-L_2(\tau_1) \given \bs{W}(\tau_1)}$.
Based on (\ref{equation:Wi evolution 1}) and (\ref{equation:Wi evolution 2}), the conditional drift between $t_k$ and $t_k+\hat{T}_k$ is
\begin{align}
\Delta &L_2(t_k, t_k+\hat{T}_k)\\
= & \E[\Big]{\bs{W}(t_k+\hat{T}_k)^T \bs{P} \bs{W}(t_k+\hat{T}_k)-\bs{W}(t_k)^T \bs{P} \bs{W}(t_k) \given \bs{W}(t_k)} \\
\leq & -2\cdot \E[\Bigg]{\bkt[\bigg]{\sum_{t=t_{k}}^{t_k+\hat{T}_k-1}{M}(t) \bm{\delta}(t)}^T\bs{P}\bs{W}(t_k)\given \bs{W}(t_k)}\label{equation:HOL WPVIS} \\ 
+& \E[\Bigg]{\bkt[\bigg]{\sum_{t=t_{k}}^{t_k+\hat{T}_k-1}M(t)\bm{\delta}(t)}^T \bs{P}\bkt[\bigg]{\sum_{t=t_{k}}^{t_k+\hat{T}_k-1}M(t)\bm{\delta}(t)} \given \bs{W}(t_k)}\label{equation:HOL VISPVIS} \\
+& 2\cdot\E[\Big]{\bkt[\big]{{\hat{T}_k\bs{1}}}^T \bs{P}{\bs{W}(t_k)} \given \bs{W}(t_k)}+\E[\Big]{{\hat{T}_k}^2 \bm{\rho}^T \bs{1}\given \bs{W}(t_k)}.\label{equation:HOL WP}
\end{align}
First, we know
\begin{equation}
\E[\Big]{\bkt[\big]{{\hat{T}_k\bs{1}}}^T \bs{P}{\bs{W}(t_k)} \given \bs{W}(t_k)}=\E[\Big]{\hat{T}_k \bm{\rho}^T \bs{W}(t_k)\given \bs{W}(t_k)}.
\end{equation} 
For any $t\geq 0$, by the assumptions on inter-arrival times and service processes, we have ${V}_i(m)\leq V_{\max}$ for every queue $i$ and $m\geq 0$, and $\bs{S}(t)\leq S_{\max}\cdot\bs{1}$, regardless of the HOL waiting time at time $t$. Therefore, (\ref{equation:HOL VISPVIS}) is bounded as
\begin{align}
&\E[\Bigg]{\bkt[\bigg]{\sum_{t=t_{k}}^{t_k+\hat{T}_k-1}M(t)\bm{\delta}(t)}^T \bs{P}\bkt[\bigg]{\sum_{t=t_{k}}^{t_k+\hat{T}_k-1}M(t)\bm{\delta}(t)} \given \bs{W}(t_k)}\label{equation:HOL VISPVIS bounded 1}\\
& \leq V_{\max}^2 S_{\max}^2 \Tr{(\bs{P})}\E[\Big]{{\hat{T}_k}^2 \given \bs{W}(t_k)}.\label{equation:HOL VISPVIS bounded 2}
\end{align}
Note that $\bs{W}(t_k)$ only tells us when the HOL job arrived and therefore provides no information about the inter-arrival times. Hence, ${V}_i(m)$ is independent of $\bs{W}(t_k)$, for any queue $i$ and $m\geq 0$.
Besides, since $\bs{S}(t)$ is also independent of the waiting time $\bs{W}(t_k)$, we can rewrite (\ref{equation:HOL WPVIS}) as
\begin{align}
& \E[\Bigg]{\bkt[\bigg]{\sum_{t=t_{k}}^{t_k+\hat{T}_k-1}{M}(t) \bm{\delta}(t)}^T\bs{P}\bs{W}(t_k)\given \bs{W}(t_k)}\\
&= \E[\Big]{\bkt[\Big]{\bkt[\big]{\hat{T}_k -T_{s}} \bkt[\Big]{{\bs{I}(t_k)} \circ \bm{\rho}^{-1}}}^T\bs{P}{\bs{W}(t_k)\given \bs{W}(t_k)}}\\
& = \E[\Big]{\bkt[\big]{\hat{T}_k-T_{s}}\cdot\bs{I}(t_k)^T\bs{W}(t_k)\given \bs{W}(t_k)},
\end{align}
where $\bm{\rho}^{-1}:=\bkt[\big]{\rho_1^{-1},\cdots,\rho_N^{-1}}$ is the vector of the reciprocal of per-queue normalized traffic load. Since the system is assumed to be stabilizable, then there exists a $J$-dimensional nonnegative vector $\bm{\beta}^T=\bkt{\beta_1,\cdots,\beta_J}$ with $\bm{\beta}^T \bs{1}<1$ such that $\sum_{j=1}^{J}\beta_j \bs{I}^{(j)}\geq \bm{\rho}$. Under the W-BMW policy, it is guaranteed that $\bs{I}(t_k)^T \bs{W}(t_k)=\max_{j:1\leq j\leq J} (\bs{I}^{(j)})^T \bs{W}(t_k)$. Therefore, we know
\begin{align}
\bm{\rho}^T\bs{W}(t_k)&\leq \bkt[\bigg]{\sum_{j=1}^{J}\beta_j\bs{I}^{(j)}}^T \bs{W}(t_k)\label{equation:HOL rho to epsilon 1}\\
&\leq \bkt[\bigg]{\sum_{j=1}^{J}\beta_j}\bs{I}(t_k)^T \bs{W}(t_k)\label{equation:HOL rho to epsilon 2} \\
&= (1-\epsilon)\bs{I}(t_k)^T \bs{W}(t_k)\label{equation:HOL rho to epsilon 3}
\end{align}
where $\epsilon:=1-\bm{\beta}^T\bs{1}$ denotes the normalized distance between the arrival rate vector and the boundary of the capacity region.
From (\ref{equation:HOL WPVIS})-(\ref{equation:HOL rho to epsilon 3}), the conditional drift can be written as
\begin{align}
\Delta L_2(t_k&, t_k+\hat{T}_k)\leq \E[\Big]{2\cdot\bkt[\big]{-\epsilon\hat{T}_k + T_{s}}\bs{I}(t_k)^T \bs{W}(t_k)+\Phi_0\hat{T}_k^{2}\given \bs{W}(t_k)},\label{equation:HOL drift bound 1}
\end{align}
where $\Phi_0:=\bkt[\big]{V_{\max}^2 S_{\max}^2\Tr\bkt[\big]{\bs{P}}+ \bm{\rho}^T \bs{1}}$ does not depend on the waiting time vector or scheduling decisions. By Lemma \ref{lemma:HOL lower bound for Tk}, we also know that
\begin{equation}
T_k\geq C_1 \bkt[\Big]{{\bs{1}}^T\bs{W}(t_k)}^{1-\alpha},
\end{equation}
where $C_1=T_s/\bkt[\big]{NK\bkt[\big]{1+(1+T_s)S_{\max}V_{\max}}}$. Here, we need to discuss two possible cases:

\vspace{3mm}
\noindent {\bf Case 1}:  $\hat{G}(\bs{W}(t_k))\geq C_1 \bkt[\Big]{{\bs{1}}^T\bs{W}(t_k)}^{1-\alpha}$
\vspace{1mm}

\noindent We first have
\begin{align}
& \bkt[\Big]{{{\bs{1}}^T\bs{W}(t_k)}}^{\alpha_1}\geq C_1{\bkt[\Big]{{\bs{1}}^T\bs{W}(t_k)}^{1-\alpha}}
\end{align}
Therefore, by the assumption that $\alpha+\alpha_1 < 1$, we have
\begin{equation}
{\bs{1}}^T\bs{W}(t_k)\leq C_1^{\frac{-1}{1-\alpha-\alpha_1}}\label{equation:HOL case 1 queue bound}
\end{equation}
Hence, from (\ref{equation:HOL drift bound 1}) and (\ref{equation:HOL case 1 queue bound}), we know that the conditional drift between $t_k$ and $t_k+\hat{T}_k$ is bounded, i.e.
\begin{align}
\Delta L_2(t_k, t_k+\hat{T}_k)&\leq \E[\Big]{2T_s\cdot\bs{I}(t_k)^T \bs{W}(t_k)+\Phi_0\hat{T}_k^{2}\given \bs{W}(t_k)} \\
& \leq  2T_{s}\cdot\bs{I}(t_k)^T \bs{W}(t_k)+ {\Phi_0\hat{G}(\bs{W}(t_k))^2} \\
& \leq  2T_{s}\cdot\bs{1}^T \bs{W}(t_k)+\Phi_0\bkt[\Big]{{{\bs{1}}^T\bs{W}(t_k)}}^{2\alpha_1}\\
& \leq  2T_{s}\cdot{C_1}^{\frac{-1}{1-\alpha-\alpha_1}}+\Phi_0  {C_1}^{\frac{-2\alpha_1}{1-\alpha-\alpha_1}}< \infty.
\end{align} 
This also implies that the unconditional drift between $t_k$ and $t_k+\hat{T}_k$ is bounded, i.e.
\begin{equation}
\E[\Big]{L_2(t_k+\hat{T}_k)-L_2(t_k)} \leq 2T_{s}\cdot{C_1}^{\frac{-1}{1-\alpha-\alpha_1}}+\Phi_0  {C_1}^{\frac{-2\alpha_1}{1-\alpha-\alpha_1}}< \infty.
\end{equation}

\noindent {\bf Case 2}:  $\hat{G}(\bs{W}(t_k))<C_1 \bkt[\Big]{{\bs{1}}^T\bs{W}(t_k)}^{1-\alpha}$
\vspace{1mm}

\noindent The above condition implies that $\hat{T}_k=\hat{G}(\bs{W}(t_k))<T_k$. Therefore, (\ref{equation:HOL drift bound 1}) can then be written as
\begin{align}
\Delta L_2(t_k, t_k+\hat{T}_k)&\leq 2\bkt[\Big]{-\epsilon\hat{G}(\bs{W}(t_k)) + T_{s}}\bs{I}(t_k)^T \bs{W}(t_k)+ {\Phi_0\hat{G}(\bs{W}(t_k))^2} \\
& \leq -\frac{2\epsilon}{N} \bkt[\Big]{{{\bs{1}}^T\bs{W}(t_k)}}^{1+\alpha_1}+2T_{s}\bkt[\Big]{{{\bs{1}}^T\bs{W}(t_k)}}+\Phi_0 \bkt[\Big]{{{\bs{1}}^T\bs{W}(t_k)}}^{2\alpha_1}.\label{equation: HOL drift bound 1 in case2}
\end{align}
Since $-\frac{2\epsilon}{N} \bkt[\big]{{{\bs{1}}^T\bs{W}(t_k)}}^{1+\alpha_1}$ is the dominating term in (\ref{equation: HOL drift bound 1 in case2}), there must exist some constant $\Phi_1>0$ such that
\begin{equation}
\Delta L_2(t_k, t_k+\hat{T}_k) \leq \Phi_1 - \frac{\epsilon}{N}\bkt[\Big]{{{\bs{1}}^T\bs{W}(t_k)}}^{1+\alpha_1}.\label{equation:HOL case2 conditional drift}
\end{equation}
Moreover, we also know that
\begin{align}
\sum_{t=t_k}^{t_k+\hat{T}_k -1} \bs{1}^T \bs{W}(t)\leq \sum_{\tau=0}^{\hat{T}_k -1}\bs{1}^T \bkt[\Big]{\bs{W}(t_k)+\tau}\leq \hat{T}_k\cdot \bs{1}^T \bs{W}(t_k)+N\hat{T}_k^{2}.
\end{align}
By taking conditional expectation, we have
\begin{align}
\E[\Bigg]{\sum_{t=t_k}^{t_k+\hat{T}_k -1} \bs{1}^T \bs{W}(t)\sgiven \bs{W}(t_k)}&\leq \E[\bigg]{\hat{T}_k\cdot \bs{1}^T \bs{W}(t_k)+N{\hat{T}_k}^2 \sgiven \bs{W}(t_k)}\\
&= \bkt[\Big]{\bs{1}^T\bs{W}(t_k)}^{1+\alpha_1}+ N\bkt[\Big]{\bs{1}^T\bs{W}(t_k)}^{2\alpha_1}\\
&\leq (1+N)\bkt[\Big]{\bs{1}^T\bs{W}(t_k)}^{1+\alpha_1}.\label{equation:HOL bound of sum queue length}
\end{align}
The last inequality holds since $\alpha_1<1$. Therefore, based on (\ref{equation:HOL case2 conditional drift}) and (\ref{equation:HOL bound of sum queue length}), we obtain that
\begin{equation}
\E[\Big]{L_2(t_k+\hat{T}_k)-L_2(t_k)} \leq \Phi_1 - \frac{\epsilon}{N_2}\E[\Bigg]{\sum_{t=t_k}^{t_k+\hat{T}_k-1}{\bs{1}^T\bs{W}(t)}},\label{equation:case2 unconditional drift bound}
\end{equation}
where $N_2:=N(N+1)$.

Next, we consider the slot-by-slot conditional drift for any $t$ between $t_k+\hat{T}_k$ and $t_{k+1}$. Note that there is no switching between $t_k+\hat{T}_k$ and $t_{k+1}$. Therefore,
\begin{align}
\Delta L_2(t,t+1)&= \E[\Big]{\bs{W}(t+1)^T \bs{P} \bs{W}(t+1)-\bs{W}(t)^T \bs{P} \bs{W}(t) \given \bs{W}(t)}\\
&\leq 2\cdot \E[\Big]{\bkt[\Big]{\bs{1}-\bm{\delta}(t)}^{T}\bs{P}\bs{W}(t) \given \bs{W}(t)}\label{equation:HOL one-slot VISPW} \\ 
&\hspace{6pt} + \E[\Big]{\bkt[\Big]{\bs{1}-\bm{\delta}(t)}^T \bs{P}\bkt[\Big]{\bs{1}-\bm{\delta}(t)}\given \bs{W}(t)},\label{equation:HOL one-slot VISPVIS}
\end{align}
Similar to (\ref{equation:HOL VISPVIS bounded 1}) and (\ref{equation:HOL VISPVIS bounded 2}), we know 
\begin{equation}
\E[\Big]{\bkt[\Big]{\bs{1}-\bm{\delta}(t)}^T \bs{P}\bkt[\Big]{\bs{1}-\bm{\delta}(t)}\given \bs{W}(t)}\leq \bkt[\big]{{S_{\max}^2 V_{\max}^2}+1}\Tr({\bs{P}})=\Phi_0.
\end{equation}
Besides, since $\bs{V}(t)$ and $\bs{S}(t)$ are independent of $\bs{W}(t)$, (\ref{equation:HOL one-slot VISPW}) can be written as
\begin{align}
&  \E[\Big]{\bkt[\Big]{\bs{1}-\bm{\delta}(t)}^{T}\bs{P}\bs{W}(t) \given \bs{W}(t)}= {\bkt[\big]{\bm{\rho}^T-\bs{I}(t)^T}\bs{W}(t)}.
\end{align}
Hence, we have
\begin{equation}
\Delta L_2(t,t+1) \leq 2\cdot {\bkt[\big]{\bm{\rho}^T-\bs{I}(t)^T}\bs{W}(t)}+\Phi_0.\label{equation:HOL one-step drift}
\end{equation}
Under the W-BMW policy, at time $t$ we must have
\begin{align}
\bs{I}(t)^T \bs{W}(t)&\geq \bkt[\big]{\bs{I}^{(j)}}^T \bs{W}(t)-\frac{\bs{I}(t)^T\bs{W}(t)}{G(\bs{W}(t_k))}, \hspace{12pt}\forall j=1,...,N
\end{align}
Along with (\ref{equation:HOL rho to epsilon 1})-(\ref{equation:HOL rho to epsilon 3}), we then have
\begin{align}
(1-\epsilon)\bs{I}(t)^T \bs{W}(t)&=\sum_{j=1}^{J}\beta_j \bs{I}(t)^T \bs{W}(t)\\
& \geq \sum_{j=1}^{J} \beta_j \bkt[\big]{\bs{I}^{(j)}}^T \bs{W}(t)-\sum_{j=1}^{J}\beta_j \frac{\bs{I}^T(t)\bs{W}(t)}{G(\bs{W}(t_k))}\\
& \geq \bm{\rho}^{T} \bs{W}(t)- (1-\epsilon)\frac{\bs{I}^T(t)\bs{W}(t)}{G(\bs{W}(t_k))}.
\end{align}
Therefore, (\ref{equation:HOL one-step drift}) can be written as
\begin{align}
\Delta L_2(t,t+1)& \leq 2 \cdot \bkt[\bigg]{-\epsilon \bs{I}(t)^T\bs{W}(t)+(1-\epsilon)\frac{\bs{I}(t)^T\bs{W}(t)}{G(\bs{W}(t_k))}}+\Phi_0\label{equation:HOL one-step drift 1}\\
&\leq 2 \cdot  \bkt[\bigg]{-\frac{\epsilon}{N} \bkt[\big]{\bs{1}(t)^T\bs{W}(t)}+(1-\epsilon)\frac{\bs{I}(t)^T\bs{W}(t)}{G(\bs{W}(t_k))}}+\Phi_0\label{equation:HOL one-step drift 2} \\
&\leq  2 \cdot  \bkt[\bigg]{-\frac{\epsilon}{N} \bkt[\big]{\bs{1}(t)^T\bs{W}(t)}+(1-\epsilon)\bkt[\Big]{{\bs{I}(t)^T\bs{W}(t)}}^{1-\alpha}}+\Phi_0\label{equation:HOL one-step drift 3}\\
&\leq  2 \cdot  \bkt[\bigg]{-\frac{\epsilon}{N} \bkt[\big]{\bs{1}(t)^T\bs{W}(t)}+(1-\epsilon)\bkt[\Big]{{\bs{1}(t)^T\bs{W}(t)}}^{1-\alpha}}+\Phi_0\label{equation:HOL one-step drift 4}
\end{align}
Since $\alpha>0$, then $-\frac{\epsilon}{N} \bs{1}(t)^T\bs{W}(t)$ is the dominating term in (\ref{equation:HOL one-step drift 4}). In other words, there must exist some constant $\Phi_2>0$ such that
\begin{align}
\Delta L_2(t,t+1)& \leq \Phi_2 - \frac{\epsilon}{N} \bkt[\Big]{\bs{1}(t)^T\bs{W}(t)}.
\end{align}
Hence, for any $t\in (t_k+\hat{T}_{k}, t_{k+1})$, we know
\begin{equation}
\E[\Big]{L(t+1)-L(t)}\leq \Phi_2- \frac{\epsilon}{N} \E[\Big]{{\bs{1}(t)^T\bs{W}(t)}}.\label{equation:HOL one-step unconditional drift bound}
\end{equation}
Now, we consider any large $\T$ and let $K_{\T}$ be the number of intervals in $[0,\T)$. Since each interval lasts for at least one slot, then $K_{\T}\leq \T$. The unconditional drift in $[0,\T)$
\begin{align}
\E[\big]{L_{2}(\T)-L_{2}(0)} &= \sum_{k=0}^{K_{\T}-1} \E[\big]{L_2(t_{k+1})-L_2(t_k)}\\
& = \sum_{k=0}^{K_{\T}-1} \bkt[\Bigg]{\E[\big]{L_2(t_k+\hat{T}_k)-L_{2}(t_k)}+\sum_{\tau=t_k+\hat{T}_k}^{t_{k+1}-1}\E[\big]{L_2(\tau+1)-L_2(\tau)}  }\\
& \leq K_{\T}\Phi_1 - \frac{\epsilon}{N_2}\sum_{k=1}^{K_{\T}-1}\bkt[\Bigg]{\E[\Bigg]{\sum_{t=t_k}^{t_k+\hat{T}_k-1}{\bs{1}^T\bs{W}(t)}}} \\
&\hspace{60pt}+ \Phi_2\T -  \frac{\epsilon}{N} \sum_{k=0}^{K_{\T}-1} \bkt[\Bigg]{\E[\Bigg]{\sum_{\tau=t_k+\hat{T}_k}^{t_{k+1}-1} {\bs{1}^T\bs{W}(t)}}}\\
&\leq  K_{\T}\Phi_1 + \Phi_2\T  - \frac{\epsilon}{N_2}\bkt[\Bigg]{\E[\Bigg]{\sum_{t=0}^{\T-1}{\bs{1}^T\bs{W}(t)}}}.
\end{align}
Since  $L_2(0)=0$ and $L_2(t)$ is nonnegative regardless of $t$, by letting $\T\rightarrow \infty$, we have
\begin{equation}
\lim_{\T\rightarrow \infty}\frac{\E[\Big]{\sum_{t=0}^{\T-1}{\bs{1}^T\bs{W}(t)}}}{\T}\leq \frac{N_2(\Phi_1+\Phi_2)}{\epsilon}
\end{equation}
Moreover, for any queue $i$ at time $t$, given the information of $W_i(t)$, we also know
\begin{align}
\E[\big]{Q_i(t)\given W_i(t)} = \begin{cases}
\lambda_i W_i(t) + 1&, \hspace{6pt} W_i(t) > 0\\
                  0&, \hspace{6pt} W_i(t) = 0
\end{cases}
\label{equation:Qi and Wi}
\end{align}
By taking the unconditional expectation of (\ref{equation:Qi and Wi}), for any $t$ we have
\begin{equation}
\E[\big]{{Q_i}(t)}\leq \lambda_i {\E[\big]{{W_i}(t)}+1}, \hspace{12pt} \forall i\in \mathcal{N}.
\end{equation} 
Hence, we can conclude that 
\begin{equation}
\lim_{\T\rightarrow \infty}\frac{\E[\Big]{\sum_{t=0}^{\T-1}{\bs{1}^T\bs{Q}(t)}}}{\T}\leq  \lambda_{\max} \bkt[\bigg]{\frac{N_2(\Phi_1+\Phi_2)}{\epsilon}}+N< \infty. 
\end{equation}
Hence, the system is strongly stable under the W-BMW policy. $\qed$

\bibliographystyle{abbrv}
\bibliography{reference}

\begin{thebibliography}{10}

\bibitem{Hanbali2008}
A.~Al~Hanbali, R.~de~Haan, R.~J. Boucherie, and J.-K. van Ommeren.
\newblock A tandem queueing model for delay analysis in disconnected ad hoc
  networks.
\newblock In {\em International Conference on Analytical and Stochastic
  Modeling Techniques and Applications}, pages 189--205, 2008.

\bibitem{Allsop1972}
R.~E. Allsop.
\newblock Estimating the traffic capacity of a signalized road junction.
\newblock {\em Transportation Research}, 6(3):245--255, 1972.

\bibitem{Armony2003}
M.~Armony and N.~Bambos.
\newblock Queueing dynamics and maximal throughput scheduling in switched
  processing systems.
\newblock {\em Queueing systems}, 44(3):209--252, 2003.

\bibitem{Celik2016}
G.~Celik, S.~C. Borst, P.~A. Whiting, and E.~Modiano.
\newblock {Dynamic Scheduling with Reconfiguration Delays}.
\newblock {\em Queueing Syst. Theory Appl.}, 83(1-2):87--129, Jun 2016.

\bibitem{Celik2012}
G.~D. Celik and E.~Modiano.
\newblock Scheduling in networks with time-varying channels and reconfiguration
  delay.
\newblock In {\em Proc. IEEE INFOCOM}, pages 990--998, March 2012.

\bibitem{Chan2016}
C.~W. Chan, M.~Armony, and N.~Bambos.
\newblock {Maximum Weight Matching with Hysteresis in Overloaded Queues with
  Setups}.
\newblock {\em Queueing Syst. Theory Appl.}, 82(3-4):315--351, Apr. 2016.

\bibitem{Chen2001}
H.~Chen and D.~D. Yao.
\newblock {\em {Fundamentals of queueing networks: Performance, asymptotics,
  and optimization}}, volume~46.
\newblock Springer, 2001.

\bibitem{David2007}
F.~M. David, J.~C. Carlyle, and R.~H. Campbell.
\newblock {Context Switch Overheads for Linux on ARM Platforms}.
\newblock In {\em Proceedings of the 2007 Workshop on Experimental Computer
  Science}, ExpCS '07, 2007.

\bibitem{Eryilmaz2012}
A.~Eryilmaz and R.~Srikant.
\newblock Asymptotically tight steady-state queue length bounds implied by
  drift conditions.
\newblock {\em Queueing Systems}, 72(3-4):311--359, 2012.

\bibitem{Eryilmaz2005}
A.~Eryilmaz, R.~Srikant, and J.~R. Perkins.
\newblock Stable scheduling policies for fading wireless channels.
\newblock {\em IEEE/ACM Transactions on Networking}, 13(2):411--424, April
  2005.

\bibitem{Fan2008}
Z.~Fan.
\newblock {Wireless networking with directional antennas for 60 GHz systems}.
\newblock In {\em 14th European Wireless Conference}, pages 1--7, June 2008.

\bibitem{Ghavami2012}
A.~Ghavami, K.~Kar, and S.~Ukkusuri.
\newblock Delay analysis of signal control policies for an isolated
  intersection.
\newblock In {\em 15th International IEEE Conference on Intelligent
  Transportation Systems}, pages 397--402, 2012.

\bibitem{Gupta2010}
G.~R. Gupta and N.~B. Shroff.
\newblock Delay analysis for wireless networks with single hop traffic and
  general interference constraints.
\newblock {\em IEEE/ACM Transactions on Networking}, 18(2):393--405, 2010.

\bibitem{Hung2008}
Y.-C. Hung and C.-C. Chang.
\newblock Dynamic scheduling for switched processing systems with substantial
  service-mode switching times.
\newblock {\em Queueing systems}, 60(1-2):87--109, 2008.

\bibitem{Kar2012}
K.~Kar, S.~Sarkar, A.~Ghavami, and X.~Luo.
\newblock Delay guarantees for throughput-optimal wireless link scheduling.
\newblock {\em IEEE Transactions on Automatic Control}, 57(11):2906--2911,
  2012.

\bibitem{Le2009}
L.~B. Le, K.~Jagannathan, and E.~Modiano.
\newblock Delay analysis of maximum weight scheduling in wireless ad hoc
  networks.
\newblock In {\em 43rd Annual Conference on Information Sciences and Systems
  (CISS)}, pages 389--394, 2009.

\bibitem{LeVine2015}
S.~Le~Vine, A.~Zolfaghari, and J.~Polak.
\newblock Autonomous cars: the tension between occupant experience and
  intersection capacity.
\newblock {\em Transportation Research Part C: Emerging Technologies},
  52:1--14, 2015.

\bibitem{Levy1990}
H.~Levy and M.~Sidi.
\newblock Polling systems: applications, modeling, and optimization.
\newblock {\em IEEE Transactions on Communications}, 38(10):1750--1760, Oct
  1990.

\bibitem{Liu2015}
X.~Liu, K.~Ma, and P.~R. Kumar.
\newblock Towards provably safe mixed transportation systems with human-driven
  and automated vehicles.
\newblock In {\em 2015 54th IEEE Conference on Decision and Control (CDC)},
  pages 4688--4694, Dec 2015.

\bibitem{Mcgarry2008}
M.~P. Mcgarry, M.~Reisslein, and M.~Maier.
\newblock Ethernet passive optical network architectures and dynamic bandwidth
  allocation algorithms.
\newblock {\em IEEE Communications Surveys Tutorials}, 10(3):46--60, 2008.

\bibitem{Navda2007}
V.~Navda, A.~P. Subramanian, K.~Dhanasekaran, A.~Timm-Giel, and S.~Das.
\newblock Mobisteer: Using steerable beam directional antenna for vehicular
  network access.
\newblock In {\em Proceedings of the 5th International Conference on Mobile
  Systems, Applications and Services}, MobiSys '07, pages 192--205, 2007.

\bibitem{Neely2009}
M.~J. Neely.
\newblock Delay analysis for maximal scheduling with flow control in wireless
  networks with bursty traffic.
\newblock {\em IEEE/ACM Transactions on Networking}, 17(4):1146--1159, 2009.

\bibitem{Neely2010}
M.~J. Neely.
\newblock Stability and capacity regions or discrete time queueing networks.
\newblock {\em arXiv preprint arXiv:1003.3396}, 2010.

\bibitem{Nitsche2014}
T.~Nitsche, C.~Cordeiro, A.~B. Flores, E.~W. Knightly, E.~Perahia, and J.~C.
  Widmer.
\newblock {IEEE 802.11ad: directional 60 GHz communication for
  multi-Gigabit-per-second Wi-Fi [Invited Paper]}.
\newblock {\em IEEE Communications Magazine}, 52(12):132--141, December 2014.

\bibitem{Niu2015}
Y.~Niu, Y.~Li, D.~Jin, L.~Su, and A.~V. Vasilakos.
\newblock {A survey of millimeter wave communications (mmWave) for 5G:
  opportunities and challenges}.
\newblock {\em Wireless Networks}, 21(8):2657--2676, 2015.

\bibitem{Perkins1989}
J.~R. Perkins and P.~R. Kumark.
\newblock Stable, distributed, real-time scheduling of flexible
  manufacturing/assembly/diassembly systems.
\newblock {\em IEEE Transactions on Automatic Control}, 34(2):139--148, Feb
  1989.

\bibitem{Sharifnia1991}
A.~Sharifnia, M.~Caramanis, and S.~B. Gershwin.
\newblock Dynamic setup scheduling and flow control in manufacturing systems.
\newblock {\em Discrete Event Dynamic Systems}, 1(2):149--175, 1991.

\bibitem{Singh2011}
S.~Singh, R.~Mudumbai, and U.~Madhow.
\newblock {Interference Analysis for Highly Directional 60-GHz Mesh Networks:
  The Case for Rethinking Medium Access Control}.
\newblock {\em IEEE/ACM Transactions on Networking}, 19(5):1513--1527, Oct
  2011.

\bibitem{Takagi1988}
H.~Takagi.
\newblock {Queuing Analysis of Polling Models}.
\newblock {\em ACM Comput. Surv.}, 20(1):5--28, Mar 1988.

\bibitem{Takagi1997}
H.~Takagi.
\newblock Queueing analysis of polling models: progress in 1990-1994.
\newblock {\em Frontiers In Queueing: Models and applications in science and
  engineering}, 7:119, 1997.

\bibitem{Tassiulas1992}
L.~Tassiulas and A.~Ephremides.
\newblock Stability properties of constrained queueing systems and scheduling
  policies for maximum throughput in multihop radio networks.
\newblock {\em IEEE Transactions on Automatic Control}, 37(12):1936--1948, Dec
  1992.

\bibitem{Tassiulas1993}
L.~Tassiulas and A.~Ephremides.
\newblock Dynamic server allocation to parallel queues with randomly varying
  connectivity.
\newblock {\em IEEE Transactions on Information Theory}, 39(2):466--478, Mar
  1993.

\bibitem{Varaiya2013}
P.~Varaiya.
\newblock Max pressure control of a network of signalized intersections.
\newblock {\em {Transportation Research Part C: Emerging Technologies}},
  36:177--195, 2013.

\bibitem{Vishnevskii2006}
V.~Vishnevskii and O.~Semenova.
\newblock Mathematical methods to study the polling systems.
\newblock {\em Automation and Remote Control}, 67(2):173--220, 2006.

\bibitem{Wongpiromsarn2012}
T.~Wongpiromsarn, T.~Uthaicharoenpong, Y.~Wang, E.~Frazzoli, and D.~Wang.
\newblock {Distributed traffic signal control for maximum network throughput}.
\newblock In {\em Proc. IEEE Conference on Intelligent Transportation Systems},
  pages 588--595, Sept 2012.

\end{thebibliography}
\end{document}